\RequirePackage[2020-02-02]{latexrelease}
\RequirePackage{amsmath}
\documentclass[11pt]{amsart}
\usepackage[
backend=biber,
style=alphabetic,
]{biblatex}
\usepackage{subfigure}
\usepackage[a4paper,left=2.5cm,right=2.5cm]{geometry}
\usepackage{bbm}
\usepackage{bm}
\usepackage{xcolor}
\usepackage{graphicx}
\usepackage{hyperref}
\usepackage{amsaddr}

\hypersetup{
	colorlinks=true,
	linkcolor=blue,
	citecolor=blue,
	pdfpagemode=FullScreen,
}
\usepackage[capitalise,noabbrev,nameinlink]{cleveref}
\makeatletter
\renewcommand{\eqref}[1]{%
	\hyperref[#1]{\textup{\tagform@{\ref*{#1}}}}%
}
\makeatother
\DeclareMathOperator{\tr}{tr}

\title{Regime-switching affine term structures}
\author{Andreas Celary$^{\dag}$, Paul Eisenberg$^{\dag *}$, Zehra Eksi$^{\dag}$}
\address{$\dag$Institute for Statistics and Mathematics, WU-University of Economics and Business, Welthandelsplatz 1, 1020, Vienna, Austria}
\email{acelary@wu.ac.at, peisenbe@wu.ac.at *Corresponding author, zehra.eksi-altay@wu.ac.at}

\theoremstyle{plain}
\newtheorem{theorem}{Theorem}[section]
\newtheorem{corollary}[theorem]{Corollary}
\newtheorem{lemma}[theorem]{Lemma}
\newtheorem{proposition}[theorem]{Proposition}

\theoremstyle{definition}

\theoremstyle{remark}
\newtheorem{remark}[theorem]{Remark}

\begin{document}


\maketitle

\begin{abstract}
We consider an HJM model setting for Markov-chain modulated forward rates. The underlying Markov chain is assumed to induce regime switches on the forward curve dynamics. Our primary focus is on the interest rate and energy futures markets. After deriving HJM-drift conditions for the two markets, we prove under the assumption of affine structure for the term structure that the forward curves are solutions to specific systems of ODEs that can be solved explicitly in many cases. This allows for a tractable model setting, and we present an algorithm for obtaining consistent forward curve models within our framework. We conclude by presenting some simple numerical examples.


\end{abstract}



\section{Introduction}
A term structure relates the value of a financial contract to its time of maturity. In this context, typical examples are the term structure of interest rates, credit spreads, energy futures, and option volatilities.

This paper focuses on modeling the term structure of interest rates and energy futures. Modeling term structure for the interest rate market is essential, as risk-free interest rates are needed for the valuation and hedging of financial instruments. Moreover, the term structure of interest rates is one of the critical variables in the macroeconomic decision-making process. On the other hand, sound modeling of the term structure of energy markets is valuable for energy market participants, e.g., producers, speculators, and risk managers. Overall, both for the interest rate and the energy market, it is desirable to build a model for all possible maturities that guarantees no arbitrage in the underlying market and fits the terminal conditions defined by the contracts. Hence, mathematically, a model should be consistent for a continuum of maturities, and this presents a challenge, as the term structure becomes a high-dimensional object, for which a proper analysis is required. 

Historically, uni-variate short (interest) rate models are among the first models in the literature (see, for instance, \cite{cox, hull, vasicek}). The uni-variate or multivariate short-rate models do not consider the entire forward rate curve but only look at an infinitesimal section at the starting point. \cite{hjm} consider the continuum of maturities and the corresponding interest rates as building blocks and introduce the Heath-Jarrow-Morton (HJM) framework for modeling the entire term structure. Accordingly, the underlying object, the dynamic term structure, can be regarded as an element of an infinite-dimensional function space (see, e.g. \cite{filipovic_phd}).  The HJM framework aims to derive conditions under which the market composed of the continuum of fundamental trading instruments becomes arbitrage-free (the so-called HJM-drift or consistency conditions). In real-world markets, however, observing all possible curve realizations is impossible. A principal component analysis reveals that $3$ to $4$ factors are sufficient to account for most of the variation in the term structure of interest rates (see, for instance, \cite[Chapter 3.4]{filipovic}). From a modeling standpoint, one considers finite-dimensional realizations (FDRs) of consistent HJM models, that is, models that agree with a finite-dimensional factor model up to a time-static, deterministic mapping.

Much work has been done on the existence of FDRs, which are consistent in the HJM sense. 
In \cite{bjoerk} and \cite{tappe4} it is determined that for essentially all volatility structures with an affine FDR, the resulting forward curve are of a quasi-exponential form. For the general existence of FDRs in continuous models, following the geometric approach using Lie algebras of Björk et al. (see \cite{bjoerk1, bjoerk_svensson, bjoerk2, bjoerk3}), Teichmann and Filipovi\'c proved that under certain assumptions on the volatility functionals, the existence of an FDR necessitates the corresponding term structure to be affine (see \cite{filip, teichmann1, teich}). For the case of models with arbitrarily small jumps in all directions, the existence of FDRs has been considered in \cite{Tappe, tappe1, tappe2}. To the best of our knowledge, in the literature there are no results regarding the case of general semimartingales. 

The term structure of interest rates and energy futures exhibit a complex dynamic and stochastic behavior. One common property of the two markets is that the prices are prone to regime switches, which can be associated with business cycle effects (see, \cite{hamilton1988} and \cite{hamilton1989}). While the regime switches in the term structure of interest rates manifest themselves in parallel to changes in the macroeconomic environment (see, \cite{ang2002regime, bansal2002term}), for the energy markets, one may attribute the regime changes to fluctuations in the energy demand and relatively less elastic supply (see, \cite{huisman2003regime}).   

This paper aims to introduce Markov modulation into the existing HJM setting, particularly in the setting of FDRs. Regarding the term structure of interest rates, Markov regime switching was first introduced in the context of short rate models (see, for instance, \cite{Elliott}). The methodology has also been adapted to the HJM setting by considering Markov modulation in the coefficients of the underlying stochastic processes (see, for instance, \cite{elliott_hjm}), and the forward rate curve itself (see \cite{elliott_hjm2}). In particular, the HJM-drift condition and behavior under measure changes have been studied. The existence of FDRs in the presence of Markov-chain-modulated forward dynamics has been studied in, e.g., \cite{elhouar}. The novelty of our paper is that, we aim to introduce Markov modulation into the forward rate curve itself. Although this requires more work, it allows for instantaneous jumps in the forward curve. This is of particular interest when one wants to model sudden reactions in interest rates due to, for instance, shifts in economic policy. 
On the other hand, to the best of our knowledge, there are no studies considering Markov chain modulation in the term structure of energy futures. We also note the theoretical framework of models with jumps occuring at predetermined (deterministic) times, which often are interpreted as times of meetings to discuss policy and rate changes. For an approach modelling the term structure of multiple yield curves, see for instance \cite{thorsten1}.

Our main theoretical contribution is to analyze the structure of the Markov-chain modulated forward curves for interest rate and energy markets under the assumption that the set of admissible forward curves lies in an affine space. While this is a necessary condition in the continuous case, the case of Markov chain-induced jumps seems to still be unresolved. In light of the results of the continuous case, as well as the tractability of affine term structures, we deem our restriction as sensible. By leveraging the fact that the market is free of arbitrage, we use the HJM-drift conditions for Markov-modulated affine term structures to characterize the form of the underlying stochastic process and, under suitable assumptions, identify the forward rate curves as solutions to particular systems of ODEs. We provide a procedure to construct such admissible models. Here, note that in the context of term structure of interest rates, our results are a generalization of \cite[Section 9.3]{filipovic}(for details, see Remark \ref{rem:toprop3}).

The remainder of the paper is structured as follows. Section 2 will introduce the model setting and prove the HJM-drift conditions (the case of interest rates has been studied in \cite{elliott_hjm2}). The main results of the paper will be provided in Section 3, where we impose the affine structure on the curves and determine how the resulting admissible curve spaces behave. We will state and prove our main theorems in which we characterize the forward curves as solutions of particular systems of ODEs,  and we present the form of the underlying stochastic processes. In Section 4, we will discuss numerical aspects of the model and present some simple examples within our framework. The Appendix contains technical tools and their proofs for the reader's convenience.
 
\section{Preliminaries}
This section presents the preliminaries that will be needed to understand the main results of the paper. We first introduce the notation and model setting and then provide no-arbitrage conditions for the energy and interest rate markets.  

\subsection{Notation}
Throughout this paper, we will denote by $(\Omega ,\mathcal{F}, (\mathcal{F}_t)_{t\geq 0},\mathbb{Q})$ the stochastic basis given by a filtered probability space with a right-continuous, complete filtration $\mathcal{F}$, and measure $\mathbb{Q}$. The measure $\mathbb{Q}$ will usually play the role of a risk-neutral measure. The Lebesgue measure will be denoted by $\mu _{\mathcal{L}}$. Let $U\subseteq\mathbb{R}^d$. We will write $\bar{U}=\text{cl}(U)$ for the topological closure of $U$. Given a function $f:X\times Y\rightarrow\mathbb{R}^d$ for some $X\subseteq\mathbb{R}^m$, and open $Y\subseteq\mathbb{R}^n$ and $d,m,n\in\mathbb{N}$, we shall denote by $\partial _{y_k}f(x,y)$ the partial derivative with respect to the coordinate $y_k$ for $k\in\lbrace 1,...,n\rbrace$. Furthermore, we will use the column vector $\nabla _yf(x,y)=(\partial _{y_1}f(x,y),...,\partial _{y_n}f(x,y))^{\top}$ for the generalised gradient with respect to the variable  $y\in\mathbb{R}^n$, and similarly, $H_yf(x,y)$ the Hessian with respect to $y\in\mathbb{R}^n$. We will use $I_d\in\mathbb{R}^{d\times d}$ for the $d\times d$ identity matrix and $\text{diag}(v)$ for $v\in\mathbb{R}^d$ to denote the operator $\text{diag}:\mathbb{R}^d\rightarrow\mathbb{R}^{d\times d}$ assigning to the vector $v$ a diagonal matrix with entries $(\text{diag}(v))_{i,i}=v_i$.

Given a stochastic process $(X_t)_{t\in\mathbb{R_+}}$ with a.s. c\`adl\`ag paths, we will denote by $\Delta X_t:=X_t-X_{t-}$ the jump size at time $t$.

We will use the canonical space $E=\lbrace e_1,...,e_n\rbrace$, where $e_j\in\mathbb{R}^n$ is the $j$-th standard unit basis vector for $j=1,...,n$, $n\in\mathbb{N}$. Let $\mathcal{D}$ and $\mathcal{H}$ be two sets, which for our purposes will either denote $\mathbb{R}^d$ or $\mathbb{R}^{d\times d}$. Given a function $f:E\rightarrow\mathcal{D}$, we define the state vector of $f$ as $\bm{f}=(f(e_1),...,f(e_n))^{\top}\in\mathcal{D}^n$. Given another function $g:E\rightarrow\mathcal{H}$ and $\bm{g}\in\mathcal{H}^n$, we will define the component-wise matrix product of the state vectors $\bm{f}\odot\bm{g}$ as $(\bm{f}\odot\bm{g})_j:=\langle f(e_j),g(e_j)\rangle$ for $j=1,...,n$. In that case, we will make frequent use of the identification $(\mathbb{R}^d)^n\cong\mathbb{R}^{d\times n}$. For $f:E\rightarrow\mathbb{R}$ and $g:E\rightarrow\mathbb{R}^d$, we can write $f(z)=\langle \bm{f},z\rangle$ and $g(z)=\bm{g}z$, where the second product is understood as a matrix multiplication.

\subsection{Setting}
Let $(\Omega ,\mathcal{F}, (\mathcal{F}_t)_{t\geq 0},\mathbb{Q})$ support a standard $d$-dimensional Brownian motion $W$ and a continuous-time finite-state Markov chain $Z$ taking values in the canonical space $E$. We denote the transition matrix of $Z$ by $Q=(q_{ij})_{i,j\in E}$ and its infinitesimal generator by $\mathcal{G}$. Furthermore, $F:E\times\mathcal{B}(E)\rightarrow [0,1]$ will denote its jump law, $\lambda :E\rightarrow \mathbb{R} _+$ its jump intensity and $\nu$ its compensator. Let $Y$ be a continuous d-dimensional It\^o process
\begin{equation*}
	Y_t=Y_0+\int _0^tb_sds+\int _0^t\sigma _sdW_s,
\end{equation*}
where $Y_0$ is an $\mathcal{F}_0$-measurable random variable in $\mathbb{R}^d$, $b:[0,T]\times\Omega\rightarrow\mathbb{R}^d$ is a d-dimensional progressively measurable process with a.s. integrable paths and $\sigma :[0,T]\times\Omega\rightarrow\mathbb{R}^{d\times d}$ is a d-dimensional matrix-valued process which is progressively measurable and has square integrable paths. Furthermore, we will require the pair $(Y,Z)$ to be Markovian on the state space $\bar{U}\times E$ for some open subset $U\subseteq\mathbb{R}^d$. Note that this essentially implies $b_t=b(Y_t,Z_t)$ and $\sigma _t=\sigma (Y_t,Z_t)$ for all $t$, $\lambda _{\mathcal{L}}\otimes\mathbb{Q}$-a.e. for suitable functions $b:\mathbb{R}^d\times E\rightarrow\mathbb{R}^d$ and $\sigma :\mathbb{R}^d\times E\rightarrow\mathbb{R}^{d\times d}$, which will we use henceforth (for details, see for instance \cite[Theorem 7.16]{jacod2}).

We consider two markets for our modeling purposes: the energy futures market and the fixed income market. In both cases, the main object of our interest will be the forward rate, denoted by $f(t,T)$. In the case of the energy futures market, the forward rate $f(t,T)$ denotes the time $t$ price of ``delivering $1$ unit of energy`` at time $T$. For the fixed income market, $f(t,T)$ denotes the instantaneous forward rate realized at time $t$ for a contract with maturity $T$, that is the interest rate contracted from time $T$ to $T+\delta$ for $\delta\rightarrow 0$ (for details, see for instance \cite{bjoerk4}). In both cases, the underlying object can be modeled using the HJM-framework and can therefore be understood as an infinite-dimensional stochastic process. We will therefore not make a distinction in notation between the two markets. 

For our purposes, we will make use of the Musiela parametrization (see, e.g. \cite{musiela}). Instead of considering a market with a continuum of maturity dates $T$, and therefore a continuum of forward rates $f(t,T)$, we will define for $x\geq0$ the function $f_t(x):=f(t,t+x)$, the forward rate at time $t$ as the function of time to maturity. This makes the forward rates an object of a suitably chosen curve space, (see, e.g. \cite{filipovic_phd}).

Let $g:[0,T]\times\mathbb{R}^d\times E\rightarrow\mathbb{R}$ be such that $g(\cdot ,\cdot, e_i)$ is of class $C^{1,2}$ for each $i=1,...,n$. We adopt the approach used in \cite{buehler} and assume the forward rate process $f_t$ is of the form
\begin{equation*}
	f_t(x)=g(x,Y_t,Z_t).
\end{equation*}
By using It\^o's Lemma for semimartingales (see, for instance \cite[Theorem I.4.57]{jacod}), we see that for fixed $x$, the forward rate process dynamics satisfy
\begin{equation*}
	f_t(x)=f_0(x)+\int _0^t\beta _s(x)ds+\int _0^t\Sigma _s(x)dW_s+\sum _{s\leq t}\Delta g(x,Y_s,Z_s),
\end{equation*}
where
\begin{equation}\label{eq:1}
	\begin{aligned}
		\beta _t(x)&:=\langle b_t,\nabla _yg(x,Y_t,Z_t)\rangle+\frac12\tr\left( \sigma_t\sigma _t^{\top}H_y(g)(x,Y_t,Z_t)\right),\\
		\Sigma _t(x)&:=\left(\nabla _yg(x,Y_t,Z_t)\right) ^{\top}\sigma _t.
	\end{aligned}
\end{equation}

\subsection{No-arbitrage conditions for the interest rate and energy markets}
We will derive no-arbitrage conditions for the interest rate and energy markets within our modeling framework in the following. Note that for the interest rate market, \cite{elliott_hjm2} introduces no-arbitrage conditions for a similar setting. We still provide our results below for the paper to be self-contained and for readers' convenience. We refer to the appendix for proofs of the statements. 
\subsubsection{Energy market}\label{sec:energy}
For the energy market, we will analyze the contract of delivering one unit of energy in a time interval $[T_1,T_2]$ for $T_1<T_2$. The forward rate $f_t(x)$ is understood as the price at time $t$ of delivering energy at time $t+x$. The value of the future contract $F(t,T_1,T_2)$ at time $t$ is defined by
\begin{equation*}
	\begin{aligned}
		F(t,T_1,T_2)&=\frac{1}{T_2-T_1}\int _{T_1-t}^{T_2-t}f_t(x)dx.
	\end{aligned}
\end{equation*}
To derive no-arbitrage conditions, we need to consider the discounted prices of contracts, that is
$\tilde{F}(t,T_1,T_2)=e^{-rt}F(t,T_1,T_2)$, where $r>0$ is a constant interest rate.
We have the following no-arbitrage condition
\begin{proposition}\label{prop:1}
The discounted future prices processes $(\tilde{F}(t,T_1,T_2)) _{t\in [0,T_1]}$ are local martingales for any $0\leq T_1<T_2<\infty$ if and only if for all $x,t\geq 0$ the following drift condition holds
        \begin{equation} \label{eq:2}
		\beta _t(x)=\partial _xg(x,Y_t,Z_t)+rg(x,Y_t,Z_t)-\sum _{j\in E}\Delta f(j;t,x) q_{Z_t,j},\qquad\mu _{\mathcal{L}}\otimes\mathbb{Q}\text{-a.e.},
        \end{equation}
        where $\Delta f(j;t,x)=g(x,Y_t,j)-g(x,Y_t,Z_t)$ for $j\in E, t,x\geq 0$.
\end{proposition}
\subsubsection{Interest rate market}\label{sec:bonds}
For the interest rate market, we will consider non-defaultable zero coupon bonds as the basic trading instruments. The value of a bond with maturity $T$ at time $t$ is given by
\begin{equation*}
	P(t,T)=\exp\left( -\int _0^{T-t}f_t(x)dx\right).
\end{equation*}
We will once again have to consider the discounted contracts, that is $\tilde{P}(t,T)=e^{-\int _0^tr_sds}P(t,T)$, where $r_t:=f_t(0)$ is the short rate. As expected, we obtain a similar structure to the standard HJM-drift condition (see, e.g. \cite{hjm}). With this, we confirm the findings in \cite{elliott_hjm2} for the case of Markov-modulated HJM-models.
\begin{proposition}\label{prop:2}
	The discounted zero-coupon bond prices processes $(\tilde{P}(t,T))_{t\in [0,T]}$ are local martingales for any $0\leq T<\infty$ if and only if for all $x,t\geq 0$ the following drift condition holds
	\begin{equation}\label{eq:4}
                \begin{aligned}
		&\beta _t(x)=\partial _xg(x,Y_t,Z_t)+\Sigma _t(x)\int _0^{x}\Sigma ^{\top}_t(s)ds-\sum _{j\in E}\delta P(j;t,x)\Delta f(j;t,x)q_{Z_t,j},\qquad\mu _{\mathcal{L}}\otimes\mathbb{Q}\text{-a.e.},
                \end{aligned}
        \end{equation}
	where $\Delta f(j;t,x):=g(x,Y_t,e_j)-g(x,Y_t,Z_t)$, and $\delta P(j;t,x):=\frac{\exp (-\int _0^{x}g(s,Y_t,j)ds)}{P(t,t+x)}$ for $j\in E, t,x\geq 0$.
\end{proposition}
\section{Admissible affine curve spaces}
In this section, we will analyze the structure of the no-arbitrage conditions for the energy and interest rate markets under the assumption that the admissible forward curves belong to an affine space. In particular, we will work with the following assumption:
\begin{itemize}
	\item[\textbf{(A1)}] There are functions $c:[0,T]\times E\rightarrow\mathbb{R}$, $c(\cdot,e_i)\in C^1([0,T],\mathbb{R})$ for each $i=1,...,n$  and $u:[0,T]\times E\rightarrow\mathbb{R}^d$, $u(\cdot,e_i)\in C^1([0,T],\mathbb{R}^d)$ for each $i=1,...,n$ such that
		\begin{equation*}
			g(x,y,z)=c(x,z)+\langle y,u(x,z)\rangle.
		\end{equation*}
\end{itemize}
Under Assumption \textbf{(A1)}, the drift and diffusion coefficients of the forward rate process are
\begin{equation}\label{eq:5}
	\begin{aligned}
		\beta _t(x)&=\langle u(x,Z_t),b_t\rangle,\\
		\Sigma _t(x)&=u(x,Z_t)^{\top}\sigma _t.
	\end{aligned}
\end{equation}
We can now use the general results from Propositions \ref{prop:1} and \ref{prop:2} to derive no-arbitrage conditions tailored to the affine case. We will look at the two markets separately.
\subsection{Energy market}
Under Assumption \textbf{(A1)}, we have the following no-arbitrage condition for energy futures:
\begin{corollary}\label{cor:2}
	Let $g$ satisfy Assumption \textbf{(A1)}. The discounted energy futures price processes $(\tilde{F}(t,T_1,T_2))_{t\in [0,T_1]}$ are local martingales for any $0\leq T_1<T_2<\infty$ if and only if for all $x,t\geq 0, \mu _{\mathcal{L}}\otimes\mathbb{Q}-\text{a.e.}$ the following drift condition holds
        \begin{equation}\label{eq:6}
                \begin{aligned}
			\langle u(x,Z_t),b_t\rangle&=\partial _xc(x,Z_t)+\langle\partial _xu(x,Z_t),Y_t\rangle +r\left(\langle u(x,Z_t),Y_t\rangle +c(x,Z_t)\right)\\
			&-\sum _{j\in E}\left\langle u(x,j)-u(x,Z_t), Y_t\right\rangle-\left( c(x,j)-c(x,Z_t)\right) q_{Z_t,j}.
                \end{aligned}
        \end{equation}
\end{corollary}
\begin{proof}
We use Proposition \ref{prop:1} together with Equation \eqref{eq:5} and grouping up the terms yields the result.
\end{proof}
In what follows, we will work under the following non-degeneracy condition.
\begin{itemize}
	\item[\textbf{(A2)}]For all $z\in E$, we have $\text{aff}\lbrace u(x,z):x\geq 0\rbrace =\mathbb{R}^d$, where $\text{aff}(A)$ for a set $A\subseteq\mathbb{R}^d$ denotes its affine hull.
\end{itemize}
This is a technical condition which ensures that the curve does not belong to a lower-dimensional subspace of our original space $\mathbb{R}^d$ (i.e. no redundant coordinates) and puts no restriction on the choice of possible models. Next, we will show that the assumptions \textbf{(A1)} , \textbf{(A2)} together with the no-arbitrage assumption impose certain structural restrictions on the driving process $Y$. Namely, the following proposition proves that under the given assumptions, the process $Y$ has an affine drift.
\begin{proposition}\label{prop:3a}
	Assume $g$ satisfies the Assumptions \textbf{(A1)} and \textbf{(A2)}. Furthermore, let $(\tilde{F}(t,T_1,T_2))_{t\in [0,T_1]}$ be a local martingale for any $0\leq T_1<T_2<\infty$. Then the process $Y_t$ has affine drift: there exist vectors $\beta _i:E\rightarrow\mathbb{R}^d$ for $i=0,...,d$, such that $b_t=\beta _0(Z_t)+\sum _{i=1}^d\beta _i(Z_t)Y_{t,i}$ for any $t\geq 0$, $\mathbb{Q}-a.s.$
\end{proposition}
\begin{proof}
        From \eqref{eq:6} we have 
        \begin{equation}\label{eq:7}
                \begin{aligned}
			\langle u(x,Z_t),b_t\rangle&=r\langle u(x,Z_t),Y_t\rangle +\langle \partial _xu(x,Z_t),Y_t\rangle -\sum _{j\in E}\langle u(x,j)-u(x,Z_t), Y_t\rangle q_{Z_t,j}\\
                                          &+rc(x,Z_t)-\sum _{j\in E}\left( c(x,j)-c(x,Z_t)\right) q_{Z_t,j}+\partial _xc(x,Z_t).
              \end{aligned}
        \end{equation}
	Fix a state $Z_t=z\in E$. It follows from \textbf{(A2)} that there are $x_1,...,x_d\in\mathbb{R}^d$, such that $\text{span}\lbrace u(x_1,z),...,u(x_d,z)\rbrace =\mathbb{R}^d$. Let $U(z)$ be the $d\times d$ matrix with entries $U_{ij}(z)=u_j(x_i,z)$, $U'(z)$ be the $d\times d$ matrix with entries $U'_{ij}(z)=\partial _xu_j(x_i,z)$, $C(z)$ be the $d$-dimensional vector with entries $C_i(z)=c(x_i,z)$, and finally, $C'(z)$ be the $d$-dimensional vector with entries $C'_i(z)=\partial _xc(x_i,z)$ for $i=1,...,d$. We can then write \eqref{eq:7} as the following system of equations
        \begin{equation*}
                \begin{aligned}
                        U(z)b_t&=rU(z)Y_t-U'(z)Y_t-\sum _{j\in E}\left( U(j)-U(z)\right) Y_tq_{z,j}\\
                               &+rC(z)+C'(z)-\sum _{j\in E}\left( C(j)-C(z)\right) q_{z,j}.
                \end{aligned}
        \end{equation*}
	The matrix $U(z)$ is invertible. We can therefore write the solution of the linear system as
        \begin{equation*}
                \begin{aligned}
                        b_t &=rY_t-U^{-1}(z)U'(z)-\sum _{j\in E}\left( U^{-1}(z)U(j)-I_d\right) Y_tq_{z,j}\\
			    &+rU^{-1}(z)C(z)+U^{-1}(z)C(z)-\sum _{j\in E}U^{-1}(z)\left( C(j)-C(z)\right) q_{z,j}.
                \end{aligned}
        \end{equation*}
	From this, we see $b_t=b(Y_t,Z_t)$ is in fact affine in $Y_t$.
\end{proof}
The next theorem shows that Assumption \textbf{(A1)} on $g$ and the no-arbitrage condition result in the curves $c$ and $u$ to be determined via a system of ordinary differential equations (ODEs).
\begin{theorem}\label{thm:1}
	Assume that $g$ satisfies Assumptions \textbf{(A1)} and \textbf{(A2)} and that $(\tilde{F}(t,T_1,T_2))_{t\in [0,T_1]}$ is a local martingale for any $0\leq T_1<T_2<\infty$. Then $b_t=\beta _0(Z_t)+\sum _{i=1}^d\beta _i(Z_t)Y_{t,i}$ for some $\beta _i:E\rightarrow\mathbb{R}^d, i=1,...,d$, and the state vectors $\bm{u}(x)$ and $\bm{c}(x)$ are the solutions of the following linear ODE
	\begin{equation}\label{eq:19}
                \begin{aligned}
                        \bm{u}'(x)&=L\bm{u}(x),\\
                        \bm{c}'(x)&=M\bm{c}(x)+\bm{\beta}_0\odot\bm{u}(x),
                \end{aligned}
        \end{equation}
        where $L$ and $M$ are linear operators defined as
	\begin{equation}\label{eq:20}
                \begin{aligned}
			(L\bm{u})_j&=\sum _{i=1}^d(\bm{\beta}_i\odot\bm{u})_je_i - r\bm{u}_j+(\bm{u}Q^{\top})_j\qquad\text{for } j=1,...,n,\\
                        M\bm{c}&=Q\bm{c}-r\bm{c}.
                \end{aligned}
        \end{equation}
\end{theorem}
\begin{proof}
	Assume $(\tilde{F}(t,T_1,T_2))_{t\in [0,T_1]}$ is a local martingale. Then Equation \eqref{eq:6} holds. Because of Proposition \ref{prop:3a}, $b_t$ is in fact an affine function of $Y_t$, i.e. $b_t=\beta _0(Z_t)+\sum _{i=1}^s\beta _i(Z_t)Y_{t,i}$ for some $\beta _i:E\rightarrow\mathbb{R}^d$, $i=0,...,d$. Conditioning Equation \eqref{eq:6} on $(Y_t,Z_t)=(y,z)$ for $(y,z)\in \bar{U}\times E$, we find
	\begin{equation}\label{eq:7a}
                \begin{aligned}
			&\left\langle \beta _0(z)+\sum _{i=1}^d\beta _i(z)y_i,u(x,z)\right\rangle -r\langle u(x,z),y\rangle-\langle\partial _xu(x,z),y\rangle +\sum _{j\in E}\langle u(x,j)-u(x,z),y\rangle q_{z,j}\\
			&=rc(x,z)+\partial _xc(x,z)-\sum _{j\in E}\left( c(x,j)-c(x,z)\right) q_{z,j}\qquad\text{for }x\geq 0.
                \end{aligned}
        \end{equation} 
        We now take the gradient with respect to the argument $y$ on both sides to obtain
	\begin{equation}\label{eq:7b}
                \begin{aligned}
			\partial _xu(x,z)=\sum _{i=1}^d\langle\beta _i(z),u(x,z)\rangle e_i-ru(x,z)+\sum _{j\in E}\left( u(x,j)-u(x,z)\right) q_{z,j}.
                \end{aligned}
        \end{equation}
	By plugging Equation \eqref{eq:7b} into Equation \eqref{eq:7a} above, we get the following system of equations
        \begin{equation*}
                \begin{aligned}
			\partial _xu(x,z)&=\sum _{i=1}^d\langle\beta _i(z),u(x,z)\rangle e_i-ru(x,z)+\sum _{j\in E}\left( u(x,j)-u(x,z)\right) q_{z,j},\\
			\langle\beta _0(z),u(x,z)\rangle&=rc(x,z)+\partial _xc(x,z)-\sum _{j\in E}\left( c(x,j)-c(x,z)\right) q_{z,j}.
                \end{aligned}
        \end{equation*} 
        This system has to hold for each state $z\in E$. We can therefore write it in terms of the state vectors to get
\begin{equation*}
                \begin{aligned}
			\bm{u}_j'(x)&=\sum _{i=1}^d(\bm{\beta}_i\odot\bm{u}(x))_je_i-r\bm{u}_j(x)+(\bm{u}(x)Q^{\top})_j,\\
			\bm{c}'(x)&=Q\bm{c}(x)-r\bm{c}(x)+\left( \langle\beta _0(e_1),u(x,e_1)\rangle,...,\langle\beta _0(e_n),u(x,e_n)\rangle\right) ^{\top},
                \end{aligned}
        \end{equation*}
        which proves the assertion.
\end{proof}
\begin{remark}\label{rem:1}
	The initial value problem specified in Theorem \ref{thm:1} is a coupled linear system of ODEs. Thus, we can specify the solutions in terms of exponential functions of the linear operators $L$ and $M$ and solve the problem numerically using standard methods (see, e.g. \cite{kloeden}).
\end{remark}
In the case of a one-dimensional Brownian motion and an $n$-state Markov chain we can solve Equation \eqref{eq:20} explicitly.
\begin{corollary}\label{cor:t1}
	Assume $d=1$ and let $b(Y_t,Z_t)=\beta _1(Z_t)+\beta _2(Z_t)Y_t$, where $\beta _i:E\rightarrow\mathbb{R}$ for $i=1,2$, and $\beta _2(e_k)\neq 0$ for all $k\in\lbrace 1,...,n\rbrace$. Then the functions $\textbf{u},\textbf{c}:[0,T]\rightarrow\mathbb{R}^n$ are given by
        \begin{equation*}
                \begin{aligned}
			\textbf{u}(x)&=\exp\left( x\left(B_2-r\mathbbm{1}+Q\right)\right)\textbf{u}_0,\\
			\textbf{c}(x)&=B_1B_2^{-1}\exp\left( x\left( B_2-r\mathbbm{1}+Q\right)\right) \textbf{c}_0,
                \end{aligned}
        \end{equation*}
	for $x\geq 0$, where $B_i:=\emph{diag}(\beta _i(e_1),...,\beta _i(e_n))$ for $i\in\lbrace 1,2\rbrace$. If $\beta _2(e_k)=0$ for some $k$, then $\textbf{c}(x)=\tilde{B}\exp\left( x\left( Q-r\mathbbm{1}+B_2\right)\right) C_0$, where $\tilde{B}=\emph{diag}\left( \frac{\beta _1(e_1)}{\beta _2(e_1)},...,\beta _1(e_k),...,\frac{\beta _1(e_n)}{\beta _2(e_n)}\right)$.
\end{corollary}
\begin{proof}
	This follows immediately from Theorem \ref{thm:1} by setting $d=1$.
\end{proof}
Additionally, one can make use of the special form of the transition matrix $Q$ by choosing a specific starting point $\textbf{u}_0$ to obtain a simplified form of the solution.
\begin{corollary}
	Let $d=1$ and $b(Y_t,Z_t)=\beta _1(Z_t)+\beta _2(Z_t)Y_t$ for $\beta _i:E\rightarrow\mathbb{R}$, $i=1,2$ and $\beta _2(e_k)\neq 0$ for all $k\in 1,...,n$. Choose $\textbf{u}_0=\lambda\sum _{i=1}^n e_i$ for some $\lambda \in\mathbb{R}$. Then
        \begin{align*}
                u(x,Z_t)=\lambda e^{x\beta _2(Z_t)}-r.
\end{align*}
\end{corollary}
\begin{proof}
	$Q$ is a transition matrix of the Markov chain $Z$ and therefore has $0$ as an eigenvalue. The vector $\lambda\sum _{i=1}^ne_i$ belongs to its kernel, i.e. $Q\textbf{u}_0=0$. We therefore get
        \begin{align*}
		u(x,Z_t)&= \langle\textbf{u}(x),Z_t\rangle=\langle\left(\exp\left( x(Q+r\mathbbm{1}+\text{diag}(\beta _2(e_1),...,\beta _2(e_n)))\right) \textbf{u}_0\right) ,Z_t\rangle \\
			&=\langle\text{diag}(e^{\beta _2(e_1)x}-r,...,e^{\beta _2(e_n)x}-r)\exp (xQ)\textbf{u}_0 ,Z_t\rangle\\
			&=\langle\text{diag}(e^{\beta _2(e_1)x}-r,...,e^{\beta _2(e_n)x}-r)\textbf{u}_0,Z_t\rangle \\
			&=\lambda\sum _{i=1}^n(e^{\beta _2(e_i)}-r)\langle e_i,Z_t\rangle =\lambda e^{x\beta _2(Z_t)}-r.
        \end{align*}
\end{proof}

\begin{remark}\label{rem:1a}
	The practitioner may now obtain a forward curve model satisfying \textbf{(A1)} and fulfilling the no-arbitrage condition via the following procedure:
	
	Given diffusion coefficients $\sigma :\mathbb{R}^d\times E\rightarrow\mathbb{R}^{d\times d}, \beta _0,...,\beta _d:E\rightarrow\mathbb{R}^d$, starting values $y_0\in\mathbb{R}^d,z_0\in E,u_0:E\rightarrow\mathbb{R}^d,c_0:E\rightarrow\mathbb{R}$, a transition matrix $Q\in\mathbb{R}^{n\times n}$, and a discount rate $r>0$:
	\begin{itemize}		
		\item Define the operators $L$ and $M$ by \eqref{eq:20}.
		\item Solve the system of linear ODEs given by \eqref{eq:19} with initial values $u(0)=u_0$ and $c(0)=c_0$.
		\item Define $g(x,y,z)=c(x,z)+\langle u(x,z),y\rangle$.
		\item Define $Z$ as the Markov chain with intensity matrix $Q$ and starting value $Z_0=z_0$.
		\item Define $Y$, an It\^o process with $Y_0=y_0$ and $dY_t=\left(\beta _0(Z_t)+\sum _{j=1}^d\beta _j(Z_t)Y_{t,i}\right) dt+\sigma (Y_t,Z_t)dW_t$.
		\item Set $f_t(x)=g(x,Y_t,Z_t)$.
	\end{itemize}
\end{remark}
\subsection{Interest rate market}
The drift condition concerning the interest rate market has an additional term that depends on the volatility. Due to that, the function $u$ feeds into the drift condition. In order to be able to derive  results in explicit form, we will work under the following simplifying assumption:
\begin{itemize}
	\item[\textbf{(A3)}] The function $u$ in Assumption \textbf{(A1)} is independent of the state of the Markov chain, that is $u(x,z)=u(x)$.
\end{itemize}
We now state the simplified no-arbitrage condition.
\begin{corollary}\label{cor:3}
	Let $g$ satisfy the Assumptions \textbf{(A1)} and \textbf{(A3)}. The discounted zero-coupon bond price processes $(\tilde{P}(t,T))_{t\in [0,T]}$ are local martingales for any $0\leq T<\infty$ if and only if for all $x,t\geq 0, \mu _{\mathcal{L}}\otimes\mathbb{Q}-\text{a.e.,}$ the following drift condition holds
        \begin{equation}\label{eq:8}
                \begin{aligned}
			&\langle u(x),b_t\rangle=\langle\partial _xu(x),Y_t\rangle+\partial _xc(x,Z_t)+\left\langle u(x),a_t\int _0^xu(s)ds\right\rangle\\-
			&\sum _{j\in E}\left( c(x,j)-c(x,Z_t)\right)\exp\left(\int _0^x c(s,Z_t)-c(s,j)\right) q_{Z_t,j},
                \end{aligned}
        \end{equation}
	where $a_t=\sigma _t\sigma _t^{\top}$.
\end{corollary}
\begin{proof}
	Proposition \ref{prop:2} along with equations \eqref{eq:5} yield the desired form.
\end{proof}
\begin{remark}\label{rem:2}
	Note that Corollary \ref{cor:3} can be extended to the case where assumption \textbf{(A3)} does not hold. This directly follows from using Proposition \ref{prop:2}.
\end{remark}
Before proceeding further, we additionally introduce the following notation:
\begin{equation*}
	\begin{aligned}
		v(x)&=\int _0^xu(s)ds,\\
		w(x,z)&=\int _0^xc(s,z)ds.
	\end{aligned}
\end{equation*}
Rearranging terms in \eqref{eq:8} and using the functions defined above, we write the drift condition as
\begin{equation}\label{eq:9}
	\begin{aligned}
		&\langle v'(x),b_t\rangle-\langle v'(x),a_tv(x)\rangle-\langle v''(x),Y_t\rangle\\
		&=\partial _{xx}w(x,Z_t)-\sum _{j\in E}\left( \partial _xw(x,j)-\partial _xw(x,Z_t)\right)\exp\left( w(x,Z_t)-w(x,j)\right) q_{Z_t,j}.
	\end{aligned}
\end{equation}
Define the function $H:[0,T]\times E\rightarrow\mathbb{R}$ as the integral of the right-hand side of \eqref{eq:9} with respect to the variable $x$, that is
\begin{equation}\label{eq:10}
	\begin{aligned}
		H(x,z):=\partial _xw(x,z)+\sum _{j\in E}\exp\left( w(x,z)-w(x,j)\right) q_{z,j}-c(0,z),
	\end{aligned}
\end{equation}
where the integration constant is chosen in such a way that $H(0,z)=0$. Conditioning equation \eqref{eq:9} on $(Y_t,Z_t)=(y,z)$ for $(y,z)\in\bar{U}\times E$ and integrating on both sides with respect to $x$, we observe that
\begin{equation}\label{eq:11}
	\begin{aligned}
		H(x,z)=\langle v(x),b(y,z)\rangle-\frac12\langle v(x),a(y,z)v(x)\rangle-\langle v'(x),y\rangle+L(y,z),
	\end{aligned}
\end{equation}
where $L(y,z)$ is an integration constant possibly depending on $y$ and $z$. We see, however, that by plugging in $x=0$ into \eqref{eq:11}, $L(y,z)=\langle v'(0),y\rangle=\langle u_0,y\rangle$ for $u(0)=u_0\in\mathbb{R}^d$. We will consider Equations \eqref{eq:10} and \eqref{eq:11} to determine the functions $v$ and $w$ separately. First, we show that under the assumptions we made, the driving process $Y$ takes a form which is almost affine, but with additional correction terms.
\begin{proposition}\label{prop:3}
	Assume $g$ satisfies \textbf{(A1),(A2),(A3)} and let $(\tilde{P}(t,T))_{t\in [0,T]}$ be a local martingale for any $0\leq T<\infty$. Then there exist functions $\Lambda _i:\mathbb{R}^d\times E\rightarrow\mathbb{R}$ and coefficients $a_i\in\mathbb{R}^{d\times d}$ and $b_i\in\mathbb{R}^d$ for $i=1,...,C_d$, where $C_d\leq (d+1)(d/2+1)$ is a constant depending only on $d$, as well as affine functions $A:\mathbb{R}^d\times E\rightarrow\mathbb{R}^{d\times d}$, $A(y,z)=A_0(z)+\sum _{i=1}^dA_i(z)y_i$ for appropriate $A_i(z):E\rightarrow\mathbb{R}^{d\times d}$, $i=0,...,d$ and $B:\mathbb{R}^d\times E\rightarrow\mathbb{R}^d$, $B(y,z)=\beta _0(z)+\sum _{i=1}^d\beta _i(z)y_i$ for appropriate $\beta _i(z):E\rightarrow\mathbb{R}^d$, $i=0,...,d$, such that the drift and diffusion coefficients of $Y$ satisfy
	\begin{equation}\label{eq:12}
		\begin{aligned}
			b(y,z)=\sum _{i=1}^{C_d}\Lambda _i(y,z)b_i+B(y,z),\\
			a(y,z)=\sum _{i=1}^{C_d}\Lambda _i(y,z)a_i+A(y,z)
		\end{aligned}
	\end{equation}
	and we have
	\begin{equation}\label{eq:13}
		\begin{aligned}
			\langle v(x),b_i\rangle -\frac12\langle v(x),a_iv(x)\rangle &=0, \qquad\text{for any }i=1,...,d\text{ and }x\in\mathbb{R}_+.\\
		\end{aligned}
	\end{equation}
\end{proposition}
\begin{proof}
	Assume for now that $b_t=b(Y_t,Z_t)$ and $a_t=a(Y_t,Z_t)$ are of class $C^2$. Taking the gradient in Equation \eqref{eq:11} twice in the variable $y$ yields
	\begin{equation*}
		\begin{aligned}
			\langle v(x),\nabla _y(\nabla _yb(y,z))\rangle -\frac12\langle v(x),\nabla _y(\nabla _ya(y,z))v(x)\rangle =0.
		\end{aligned}
	\end{equation*}
	We may consider the above equation coordinate-wise, that is
	\begin{equation}\label{coord_wise_poly}
		\langle v(x),\partial _{y_i}\partial _{y_j}b(y,z)\rangle -\frac12\langle v(x),\partial _{y_i}\partial _{y_j}a(y,z)v(x)\rangle = 0\qquad\text{for all }i,j=1,...,d.
	\end{equation}
	This is a quadratic polynomial in the variable $v(x)$ which evaluates to zero for all $x\in\mathbb{R}_+$. Let $\mathbb{R}[X_1,...,X_d]$ denote the polynomial ring of $d$ variables over the field $\mathbb{R}$ and let $Q_V$ be the set of all quadratic polynomials vanishing on the set $V:=\lbrace v(x),x\in\mathbb{R}_+\rbrace$, that is 
	\begin{equation*}
		Q_V=\lbrace q\in\mathbb{R}[X_1,...,X_d], \text{deg}(q)\leq 2: q(v)=0 \text{ for all } v\in V\rbrace .
	\end{equation*}
	Denote the vector space dimension of $Q_V$ as $C_d:=\text{dim}(Q_V)$. We have the folllowing upper bound: $C_d\leq\text{dim}(\mathcal{S}_d(\mathbb{R}))+\text{dim}(\mathbb{R}^d)+\text{dim}(\mathbb{R})=\frac{d(d+1)}{2}+d+1=(d+1)(d/2+1)$. Therefore, any quadratic polynomial vanishing on $v$ can be written as a linear combination of no more than $C_d$ basis elements of $Q_V$. Since Equation \eqref{coord_wise_poly} holds for all $i,j=1,...,d$, we now know  that there are functions $\lambda _i:\mathbb{R}^{d\times d}\times E\rightarrow\mathbb{R}$ for $i=1,...,C_d$, such that
	\begin{equation}\label{eq:14a}
		\langle \cdot , (\nabla _y(\nabla _{y}b(y,z))\rangle-\frac12\langle\cdot ,\nabla _y(\nabla_{y}a(y,z))\cdot\rangle =\sum _{i=1}^{C_d}\lambda (y,z)q_i(\cdot ).
	\end{equation}
	Matching the linear and quadratic parts with the left-hand side in \eqref{eq:14a} yields Equation \eqref{eq:13}. Integrating \eqref{eq:14a} twice and setting $\nabla _y(\nabla_y\Lambda _i(y,z))=\lambda _i(y,z)$ for $i=1,...,d$, we obtain \eqref{eq:12}.

	Now, in order to handle the general case, we may replace every instance of the classical second derivative in the steps above with the second difference quotient $D^2_{\varepsilon}$ defined as
	\begin{equation*}
		(D^2_{\varepsilon}f(x))_{i,j=1,...,d}:=\frac{1}{\varepsilon ^2}\bigg( f(x+\varepsilon (e_i+e_j)-f(x+\varepsilon e_i)-f(x+\varepsilon e_j)+f(x)\bigg) _{i,j=1,...,d}.
	\end{equation*}
	With this we have
	\begin{equation*}
		\langle v(x), D^2_{\varepsilon}b(y,z)\rangle -\frac12\langle v(x),D^2_{\varepsilon}a(y,z)v(x)\rangle = 0,
	\end{equation*}
	where, to obtain that the right-hand side vanishes, we made use of Lemma \ref{a5}. This relaxes the assumptions on the smoothness of the drift and diffusion coefficients.
\end{proof}
\begin{remark}\label{rem:toprop3}
	The results of Proposition \ref{prop:3} are an extension of the established results found, for instance, in \cite[Section 9.3]{filipovic}, \cite{filipovic2}, and \cite{duffie} in two ways. Firstly, we extend the model through Markov chain modulation. Secondly, we relax the assumption on the linear independence of the pairwise products of the components of $v$. Note that this is a true extension, since not all possible models with affine diffusion and drift coefficient are covered under this assumption.
\end{remark}
The following results make use of Proposition \ref{prop:3} and our previous considerations to give a complete description of the form of the curve $g$. We begin by showing that the function $v$ is a solution to a particular system of Riccati ODEs. 
\begin{theorem}\label{thm:2}
	Assume $g$ satisfies Assumptions \textbf{(A1),(A2),(A3)} and let $(\tilde{P}(t,T))_{t\in [0,T]}$ be a local martingale for any $0\leq T<\infty$. Then $v$ is the solution to the system of Riccati ODEs
        \begin{equation}\label{eq:14}
		\begin{aligned}
			v'(x)-u_0&=\sum _{i=1}^d\langle\beta _i(z),v(x)\rangle e_i-\frac12\sum _{i=1}^d\langle v(x),A_i(z)v(x)\rangle e_i,\\
			v(0)&=0,
		\end{aligned}
        \end{equation}
	where $\beta _1(z),...,\beta _d(z)$ and $A_1(z),...,A_d(z)$ are as in Proposition \ref{prop:3}.
        Furthermore, the function $H(x,z)$ as defined in \eqref{eq:10} fulfills
        \begin{equation}\label{eq:15}
		H(x,z)=\langle v(x),\beta _0(z)\rangle-\frac12\langle v(x),A_0(z)v(x)\rangle,
        \end{equation}
	where $\beta_0(z)$ and $A_0(z)$ are as in Proposition \ref{prop:3}.
\end{theorem}
\begin{proof}
	Differentiating Equation \eqref{eq:11} with respect to $y$ yields
        \begin{equation*}
		\langle v(x),\nabla _yb(y,z)\rangle-\frac12 \langle v(x),\nabla _ya(y,z)v(x)\rangle- v'(x)+u_0=0\qquad\text{for }(y,z)\in\bar{U}\times E,x\geq 0.
        \end{equation*}
	We use the notation of Proposition \ref{prop:3}. Equation \eqref{eq:12} yields 
        \begin{equation*}
                \begin{aligned}
			\langle\cdot,\nabla _yb(y,z)\rangle&=\sum _{i=1}^{C_d}\nabla _y\Lambda (y,z)\langle b_i,\cdot\rangle+\sum _{i=1}^d\langle\beta _i(z),\cdot\rangle e_i,\\
			\langle\cdot ,\nabla _ya(y,z)\cdot\rangle&=\sum _{i=1}^{C_d}\nabla _y\Lambda (y,z)\langle\cdot ,a_i\cdot\rangle+\sum _{i=1}^d\langle\cdot ,A_i\cdot\rangle e_i\qquad\text{for }(y,z)\in\bar{U}\times E.
                \end{aligned}
        \end{equation*}
	Using Equation \eqref{eq:13} of Proposition \ref{prop:3}, we find
	\begin{equation*}
		\begin{aligned}
			&\langle v(x),\nabla _yb(y,z)\rangle +\frac12 \langle v(x),\nabla _ya(y,z)v(x)\rangle\\
			&=\sum _{i=1}^d\langle \beta _i(z),v(x)\rangle e_i-\frac12\sum _{i=1}^d\langle v(x),A_iv(x)\rangle e_i\qquad\text{for }(y,z)\in\bar{U}\times E, x\geq 0,
		\end{aligned}
        \end{equation*} 
	Hence, we find
	\begin{equation}\label{eq:15a}
		\begin{aligned}
			\sum _{i=1}^d\langle\beta _i(z),v(x)\rangle e_i-\frac12\sum _{i=1}^d\langle v(x),A_iv(x)\rangle e_i=v'(x)-u_0
		\end{aligned}
	\end{equation}
	as asserted.

        For the second claim, we recall once for Equation \eqref{eq:11}
        \begin{equation*}
		H(x,z)=\langle v(x),b(y,z)\rangle-\frac12 \langle v(x),a(y,z)v(x)\rangle -\langle v'(x),y\rangle+\langle u_0,y\rangle.
        \end{equation*} 
	By Proposition \ref{prop:3}, we obtain
        \begin{equation*}
                \begin{aligned}
			H(x,z)&=\left\langle v(x),\sum _{i=1}^{C_d}\Lambda (y,z)b_i+B(y,z)\right\rangle-\frac12 \left\langle v(x),\left(\sum _{i=1}^{C_d}\Lambda (y,z)a_i+A(y,z)\right) v(x)\right\rangle-\langle v'(x),y\rangle\\
			      &=\langle v(x),\beta _0(z)\rangle-\frac12 \langle v(x),A_0(z)v(x)\rangle+\left\langle \sum _{i=1}^d\langle\beta _i(z),v(x)\rangle e_i-\frac12\sum _{i=1    }^d\langle v(x),A_iv(x)\rangle e_i-v'(x)+u_0, y\right\rangle\\
			      &+\left\langle v(x),\sum _{i=1}^{C_d}\Lambda (y,z)b_i\right\rangle-\frac12 \left\langle v(x),\left(\sum _{i=1}^{C_d}\Lambda (y,z)a_i\right) v(x)\right\rangle\qquad\text{for }(y,z)\in\bar{U}\times E, x\geq 0.
                \end{aligned}
        \end{equation*}
	From Equation \eqref{eq:15a} we find
	\begin{equation*}
		\begin{aligned}
			\sum _{i=1}^d\langle\beta _i(z),v(x)\rangle e_i-\frac12\sum _{i=1    }^d\langle v(x),A_iv(x)\rangle e_i-v'(x)+u_0=0\qquad\text{for }(y,z)\in\bar{U}\times E, x\geq 0
		\end{aligned}
	\end{equation*}
	and we have
	\begin{equation*}
		\begin{aligned}
			\left\langle v(x),\sum _{i=1}^{C_d}\Lambda (y,z)b_i\right\rangle -\frac12 \left\langle v(x),\left(\sum _{i=1}^{C_d}\Lambda (y,z)a_i\right) v(x)\right\rangle =0\qquad\text{for }(y,z)\in\bar{U}\times E, x\geq 0
		\end{aligned}
	\end{equation*}
	from Equation \eqref{eq:13} of Proposition \ref{prop:3}. The claim follows.
\end{proof}
In the following, we discuss special cases under which more expliicit results can be obtained. In particular, we obtain specific values for the upper bound on the vector space dimension of the set of quadratic polynomials vanishing along the curve $v$.
\begin{proposition}\label{prop:c_d}
	Let $d$ denote the dimension of the driving process $Y$ and $C_d$ be defined as in Proposition \ref{prop:3}. Then $C_1=0$ and $C_2=1$.
\end{proposition}
\begin{proof}
	Let $Q_V=\lbrace q\in\mathbb{R}[X_1,...,X_d], \text{deg}(q)\leq 2: q(v)=0 \text{ for all } v\in V\rbrace$ and $\lbrace q_1,...,q_N\rbrace\subseteq Q_V$ be a maximal algebraically independent family of polynomials from $Q_V$. Note that this implies $N\leq d$
	Take now $V:=\lbrace v(x), x\in\mathbb{R}_+\rbrace$. Because of Assumption \textbf{(A2)}, any non-zero polynomial $q\in Q_V$ is of degree $2$, that is $\text{deg}(q)=2$ for all $q\in Q_V$ such that $q\not\equiv 0$. This follows from the following observation: assume to the contrary that there is $0\not\equiv q\in Q_V$ with $\text{deq}(q)=1$. Since $q(v(x))=0$ for $x\in\mathbb{R}_+$, the affine hull of the image of $q$ also vanishes along $v$. But since Assumption \textbf{(A2)} holds, this implies that $q\equiv 0$ on $\mathbb{R}^d$, a contradiction. Therefore, the polynomials $q\in Q_V$ are irreducible over $\mathbb{R}$. This implies the ideal $I(q_1,...,q_N)$ is prime. Consider now the algebraic variety generated by $\lbrace q_1,...,q_N\rbrace$:
	\begin{equation*}
		Z(q_1,...,q_N)=\lbrace y\in\mathbb{R}^d: q_i(y)=0\text{ for any }i=1,...,N\rbrace.
	\end{equation*}
	It is a standard result (see, for instance \cite[Chapter 1, Proposition 1.7]{hartshorne}) that 
	\begin{equation*}
		\text{dim}(Z(q_1,...,q_N))\leq\text{dim}_K\left( \mathbb{R}[y_1,...,y_d]/I(Z(q_1,...,q_N))\right),
	\end{equation*}
	where $\text{dim}_K(R)$ denotes the Krull-dimension of a ring $R$ and  $I(Z)$ denotes the ideal generated by $Z$. Note here that instead of an equality, as in \cite[Chapter 1, Proposition 1.7]{hartshorne}, we have an inequality due to the fact that $\mathbb{R}$ is not algebraically closed. Since $\mathbb{R}[X_1,...,X_d]$ is an algebra finitely generated by the $d$ monomials $y_1,...,y_d$ and $I(Z(q_1,...,q_N))$ is a prime ideal generated by the polynomials $q_1,...,q_N$, we additionally have (see \cite[Chapter 1, Theorem 1.8A]{hartshorne}) 
	\begin{equation*}
		\text{dim}_K\left( \mathbb{R}[X_1,...,X_d]/I(Z(q_1,...,q_N))\right) = d-N.
	\end{equation*}
	Now, since $v$ is by assumption non-constant, we have the following chain of inequalities 
	\begin{equation*}
		1\leq\text{dim}(Z(q_1,...,q_N))\leq d-N.
	\end{equation*}
	This gives $N\leq d-1$ which shows $\text{dim}_K(Q_V)\leq d-1$. This implies any polynomial $p\in Q_V$ can be written as a $p=\sum _{i=1}^{d-1}r_iq_i$, where $r_i\in\mathbb{R}[X_1,...,X_d]$. Now, if $d=1$, this implies $p=0$ and therefore $C_d=0$. If $d=2$, any $p\in Q_V$ is of the form $p=rq$. Since $\text{deg}(p)=2$ and $\text{deg}(q)=2$, this implies $\text{deg}(r)=0$ and therefore $C_d=\text{dim}_K(Q_V)=1$. 

\end{proof}
By setting $d=1$ we may compute the solution in a straightforward manner.  
\begin{corollary}\label{cor:th211}
        Assume $d=1$. Then the function $u$ is given by
        \begin{equation*}
                \begin{aligned}
			u(x)=\begin{cases}u_0e^{\beta _1x}, &\text{if } A_1=0,\\
                        \frac{-2}{A_1(x-\frac{1}{2\beta _1})^2}, &\text{if }\beta _1^2=-2u_0A_1,\\
                        \gamma\frac{\beta _1+\gamma}{A_1}e^{\gamma x}\frac{\frac{\beta _1+\gamma}{\beta _1-\gamma}-1}{\left( \frac{\beta _1+\gamma}{\beta  _1-\gamma}e^{\gamma x}-1\right) ^2} &\text{otherwise,}
		\end{cases}
                \end{aligned}
        \end{equation*}
        where $\gamma :=\sqrt{\beta _1^2-2u_0A_1}$.
\end{corollary}
\begin{proof}
	By the results of Proposition \ref{prop:c_d} we have that for $d=1$, the constant $C_d=0$. We additionally note that the equation
	\begin{equation*}
		v'(x)-u_0=\beta _1(z)v(x)-\frac12 A_1(z)v(x)^2
	\end{equation*}
	implies that $\beta _1$ and $A_0$ do not depend on $z$, since the function $v$ only depends on $x$ and is non-constant. Furthermore, in the case of $d=1$, the system of Riccati ODEs \eqref{eq:14} is solvable explicitly with solutions as asserted. 
\end{proof}
After obtaining the form of the function $u$ and $H$, we may now go back to Equation \eqref{eq:10} and proceed to derive solutions for the function $c$, thus finishing our derivation of the complete form of $g$.
\begin{theorem}\label{thm:3}
	Assume $g$ fulfills assumption \textbf{(A1),(A2),(A3)} and let $(\tilde{P}(t,T))_{t\in [0,T]}$ be a local martingale for any $0\leq T<\infty$. Then we have
	\begin{equation}\label{eq:thm3}
		c(x,e_i)=\frac{-\langle\bm{\tilde{w}}'(x),e_i\rangle}{\langle\bm{\tilde{w}}(x),e_i\rangle}\qquad\text{for }x\geq 0, i=1,...,n,
        \end{equation}
        where $\bm{\tilde{w}}$ is the solution to the initial value problem
	\begin{equation}\label{eq:15b}
                \begin{aligned}
			\bm{\tilde{w}}'(x)&=\left(Q-C_0+\emph{diag}\left( H(x,e_1),...,H(x,e_n)\right)\right)\bm{\tilde{w}}(x),\\
                        \bm{\tilde{w}}(0)&=(1,...1)^{\top},
                \end{aligned}
        \end{equation}
        where $C_0=\emph{diag}\left( c(0,e_1),...,c(0,e_n)\right) ^{\top}$.
\end{theorem}
\begin{remark}\label{rem:3a}
	Note that $\bm{\tilde{w}}_j(x)>0$ for all $x\in [0,T-t]$ and $j\in\lbrace 1,...,n\rbrace$. This is guaranteed since $\bm{\tilde{w}}_j(0)>0$ for all $j\in\lbrace 1,...,n\rbrace$. The solution never leaves the set $\lbrace v\in\mathbb{R}^n:v_j>0 \text{ for all } j\in \{1,...,n\}\rbrace$ by the following observation:\\
	Let $\bm{\tilde{w}}$ be a solution with $\bm{\tilde{w}}_j(0)>0$ for all $j\in\{ 1,...,n\}$. Take $x>0$, such that, $\bm{\tilde{w}}_j(y)>0$ for all $j$ and $0<y<x$. Write
        \begin{align*}
                \frac{d}{dx}\bm{\tilde{w}}_j(x)&=(Q_{jj}-c(0,j)+H(x,j))\bm{\tilde{w}}_j(x)+\sum _{i\neq j}Q_{ji}\bm{\tilde{w}}_i(x).
        \end{align*}
        Therefore,
        \begin{align*}
                \frac{d}{dx}\bm{\tilde{w}}_j(x)\geq (Q_{jj}-c(0,j)+H(x,j))\bm{\tilde{w}}_j(x),
        \end{align*}
        and by Grönwall's lemma,
        \begin{align*}
		\frac{d}{dx}\bm{\tilde{w}}_j(x)\geq\exp (Q_{jj}-c(0,j)+H(x,j))\bm{\tilde{w}}_j(0)>0.
        \end{align*}
\end{remark}
\begin{proof}[Proof of Theorem \ref{thm:3}]
	We reorder terms in Equation \eqref{eq:10} to obtain
        \begin{equation*}
		\partial _xw(x,z)=-\sum _{j\in E}\exp\left( w(x,z)-w(x,j)\right) q_{z,j}-H(x,z)+c(0,z).
        \end{equation*}
        Define now
        \begin{align*}
                \tilde{w}(x,z)=\exp (-w(x,z)).
        \end{align*}
        Then $\partial _x\tilde{w}(x,z)=-\partial _xw(x,z)\tilde{w}(x,z)$. We therefore have
        \begin{align*}
                \partial _xw(x,z)=-\sum _{j\in E}\frac{\tilde{w}(x,j)}{\tilde{w}(x,z)}q_{z,j}-H(x,z)+c(0,z).
        \end{align*}
        This can be rewritten as
        \begin{align*}
                \partial _x\tilde{w}(x,z)=\sum _{j\in E}\tilde{w}(x,j)q_{z,j}-\tilde{w}(x,z)(c(0,z)-H(x,z)),
        \end{align*}
	which holds for all values $z\in E$. Define $\bm{\tilde{w}}(x)=(\tilde{w}(x,e_1),...,\tilde{w}(x,e_n))$. Then we have the following system of ordinary differential equations
        \begin{align*}
		\bm{\tilde{w}}'(x)&=\left( Q+\text{diag}(H(x,e_1)-c(0,e_1),...,H(x,e_n)-c(0,e_n))\right)\bm{\tilde{w}}(x).
        \end{align*}
        This proves the form asserted in the theorem.
\end{proof}
In the special case that $\beta (z)\equiv\beta$ and $A_0(z)\equiv A_0$, i.e. $H(x,z)=H(z)$ we can derive a closed-form solution for the curve $c$.
\begin{corollary}\label{cor:th212}
	Assume additionally that $H(x,z)=H(x)$. Then the initial value problem given in Equation\eqref{eq:15b} has the solution
        \begin{equation*}
		\tilde{\textbf{w}}(x)=\exp\left( \int _0^x H(s)ds\right)\exp\left( x(Q-C_0)\right) (1,...,1)^{\top},
        \end{equation*}
	where $C_0$ is defined as in Theorem \ref{thm:3}.
\end{corollary}

\begin{remark}\label{rem:3b}
	The practitioner may now obtain a forward curve model satisfying \textbf{(A1)} and fulfilling the no-arbitrage condition via the following procedure:

	Given coefficients $\beta _0,...,\beta _d:E\rightarrow\mathbb{R}^d,A_0,...,A_d:E\rightarrow\mathbb{R}^{d\times d}$ for the affine part of the diffusion, functions $\Lambda _1,...,\Lambda _{C_d}:\mathbb{R}^d\times E\rightarrow\mathbb{R}$, starting values $y_0\in\mathbb{R}^d, z_0\in E, u_0\in\mathbb{R}^d,c_0:E\rightarrow\mathbb{R}$, and an intensity matrix $Q\in\mathbb{R}^{n\times n}$.
	\begin{itemize}		
		\item Solve the system of Riccati ODEs specified by \eqref{eq:14} for $v$.
		\item Define the function $H$ as in \eqref{eq:15}.
		\item Solve the initial value problem specified by \eqref{eq:15b} for $\bm{\tilde{w}}$.
		\item Obtain $u$ through $u(x)=v'(x)$ and $c$ through \eqref{eq:thm3}.
		\item Define $g(x,y,z)=c(x,z)+\langle u(x),y\rangle$.
		\item Determine coefficients $b_1,...,b_{C_d},a_1,...,a_{C_d}$ consistent with Proposition \ref{prop:3}.
		\item Set $a(y,z)=\sum _{i=1}^{C_d}\Lambda _i(y,z)a_i+A_0(z)+\sum _{i=1}^dA_i(z)y_i$ and $b(y,z)=\sum _{i=1}^{C_d}\Lambda _i(y,z)b_i+\beta _0(z)+\sum _{i=1}^d\beta _i(z)y_i$, and pick $\sigma :\mathbb{R}^d\times E\rightarrow\mathbb{R}^{d\times d}$, such that $a=\sigma\sigma^{\top}$.
		\item Define $Z$ to be the Markov chain with intensity matrix $Q$ and starting value $Z_0=z_0$.
		\item Define $Y$ as the It\^o process with $Y_0=y_0$ and $dY_t=b(Y_t,Z_t)dt+\sigma (Y_t,Z_t)dW_t$.
		\item Set $f_t(x)=g(x,Y_t,Z_t)$.
	\end{itemize}
\end{remark}
\section{Numerical discussion}
In this section, we will hold a discussion about the estimation and implementation of our model and finish the exposition with some example curves. 
\subsection{Model estimation}
The main focus of the current paper is on the introduction and theoretical derivation of the Markov-modulated setting. Algorithmic steps for model construction, which require the knowledge of model parameters, have been provided in the previous sections. The unknown parameters need to be appropriately estimated to bring the theoretical model setting to a real-world application. The choice of an appropriate method for model estimation depends on the assumption on the level of information available in the market.

Suppose one assumes the price of traded assets as well as the state of the Markov chain as observables in the market. In that case, a relevant estimation methodology should be such that it infers information about the diffusion process, which is used to compute the form of the forward curve $g$. Recently, \cite{filipovic3} introduced a kernel-based machine learning approach to estimate the discount curve. They prove that any model fitting the HJM-framework fall under the scope of their methodology. Using their approach, one may learn the forward curve by solving a Ridge regression and write down the solution to the minimization problem as a linear combination of kernel functions in a carefully chosen replicating kernel Hilbert space (RKHS). The kernels chosen in the study of \cite{filipovic3}, however, do not lie in the solution space of our Riccati equation and, therefore cannot be immediately used to infer the parameters for our model. Since HJM-compatible models fall under the scope of the curve spaces introduced in \cite{filipovic3}, by finding appropriate kernel functions, we may estimate the parameters of the Riccati equation and thereby calibrate the model to the market.
 
In the term structure of interest rates literature, the common practice assumes some factor structure, such as assuming the multi-dimensional diffusion process to represent the yields at selected fixed maturities, see, \cite{de1999dynamics}, \cite{karoui2000role} and \cite{frachot1992factor} for the main idea and specific applications. Given a model choice, availability of the information on the state of the Markov chain, and additional restrictions on the parameters, we may also address the model estimation by following the yield-factor approach. Here, note that the Markov chain may be considered to be the state variable representing the central bank policy rate. On the other hand, the estimation problem is more involved when the Markov chain is assumed to be unobservable. We leave this and all other estimation-related discussions to a future paper.

\subsection{Model implementation}
While following the algorithm outlined in \ref{rem:1a} to obtain an energy futures model appears to be straightforward, the case for the interest rate market seems to be slightly more nuanced. We refer here specifically to the result of Theorem \ref{prop:3} and the observations in Remark \ref{rem:toprop3}. In particular, Theorem \ref{prop:3} admits models of the affine class, which in practice is often sufficient. The additional degrees of freedom offered through the addition of the functions $\Lambda_i$ may be interesting from a practitioner's point of view to fit any term structure observed. The observations outlined in Remark \ref{rem:toprop3} also hold interesting implications. To be more specific, we present here a simple case of an affine model which falls under the scope of our setting, but does not fulfill the conditions of \cite[Section 9.3]{filipovic} and \cite{duffie}.

We choose the following $2$-dimensional diffusion:
\begin{equation*}
	\begin{aligned}
		dY_t&=BY_tdt+\Sigma dW_t,
	\end{aligned}
\end{equation*}
where 
\begin{equation*}
	\begin{aligned}
		B = \begin{pmatrix} 0 & -1 \\ 1 & 0 \end{pmatrix}, \text{ and } \Sigma = \begin{pmatrix} 1 & 0 \\ 0 & 1 \end{pmatrix}.
	\end{aligned}
\end{equation*}
The associated system of Ricatti ODEs reduces to a standard linear system of ODEs. We set $u_0 = e_1$ and obtain the solution $v(x)=-\sin (x)e_1+(1-\cos (x))e_2$. To check whether the model is consistent with the standard theory as described by \cite[Section 9.3]{filipovic}, \cite{filipovic2}, and \cite{duffie}, we check the linear independence of the collection $\lbrace v_1,v_2,v_1^2,v_2^2,v_1v_2\rbrace$. It follows that
\begin{equation*}
	\begin{aligned}
		\lbrace v_1,v_2,v_1^2,v_2^2,v_1v_2\rbrace = \lbrace -\sin ,1-\cos,\sin ^2,(1-\cos )^2, \sin (1-\cos )\rbrace
	\end{aligned}
\end{equation*}
is not linearly independent, since $v_2^2 = 2v_2+v_1^2$.
\subsection{Example curves}
For ease of presentation, we will restrict ourselves to the case $n=2$, that is we will make use of a $2$-state Markov chain, and $d=2$, i.e. we will use two independent driving Brownian motion for both our energy futures and interest rate models. More specifically, in the case of both markets, we will consider the following $2$-dimensional process $Y=(Y_1,Y_2)$:
\begin{equation}\label{eq:numerics1}
	\begin{aligned}
		dY_{t,1}&=\kappa _1(Y_{t,2}-Y_{t,1})dt+\varsigma _1\sqrt{Y_{t,1}}dW_{t,1},\\
		dY_{t,2}&=\kappa _2(\theta (Z_t)-Y_{t,2}dt+\varsigma _2\sqrt{Y_{t,2}}dW_{t,2}.
	\end{aligned}
\end{equation}
We observe that Equation \eqref{eq:numerics1} is consistent with the model assumptions outlined in Remark \ref{rem:1a} and \ref{rem:3b} (wherein we set $\Lambda _1(y,z)\equiv 0$) with $b(y,z)=\beta _0(z)+\beta _1y_1+\beta _2y_2$, where
\begin{equation*}
	\begin{aligned}
		\beta _0(Z_t) = \begin{pmatrix} 0 \\ \kappa _2\theta (Z_t) \end{pmatrix}, \beta _1 = \begin{pmatrix} -\kappa _1 \\ 0 \end{pmatrix}, \beta _2 = \begin{pmatrix} \kappa _1 \\ -\kappa _2 \end{pmatrix},
	\end{aligned}
\end{equation*}
and $\sigma (y,z)=\sigma _0+\sigma _1\sqrt{y_1}+\sigma _2\sqrt{y_2}$, where
\begin{equation*}
	\begin{aligned}
		\sigma _0 = 0, \sigma _1=\begin{pmatrix} \varsigma _1 & 0 \\ 0 & 0 \end{pmatrix}, \sigma _2=\begin{pmatrix} 0 & 0 \\ 0 & \varsigma _2 \end{pmatrix}.
	\end{aligned}
\end{equation*}
We set the starting value $y_0=\begin{pmatrix} 0.1 \\ 0.2 \end{pmatrix}$ and make the following choices for the numerical values of the parameters of the diffusion
\begin{equation*}
	\begin{aligned}
		\kappa _1&=0.9, \kappa _2=0.5,\\
		\varsigma _1&= 0.2, \varsigma _2 = 0.15,\\
		\theta (e_1)&=0.6, \theta (e_2) = 0.1.
	\end{aligned}
\end{equation*}
For the Markov chain, we set the starting value $z_0=e_1$ and transition matrix $Q=\begin{pmatrix} -1 & 1 \\ 2 & -2 \end{pmatrix}$.
\begin{figure}[!h]
	\begin{center}
		\subfigure{\includegraphics[width=16cm, height = 8cm]{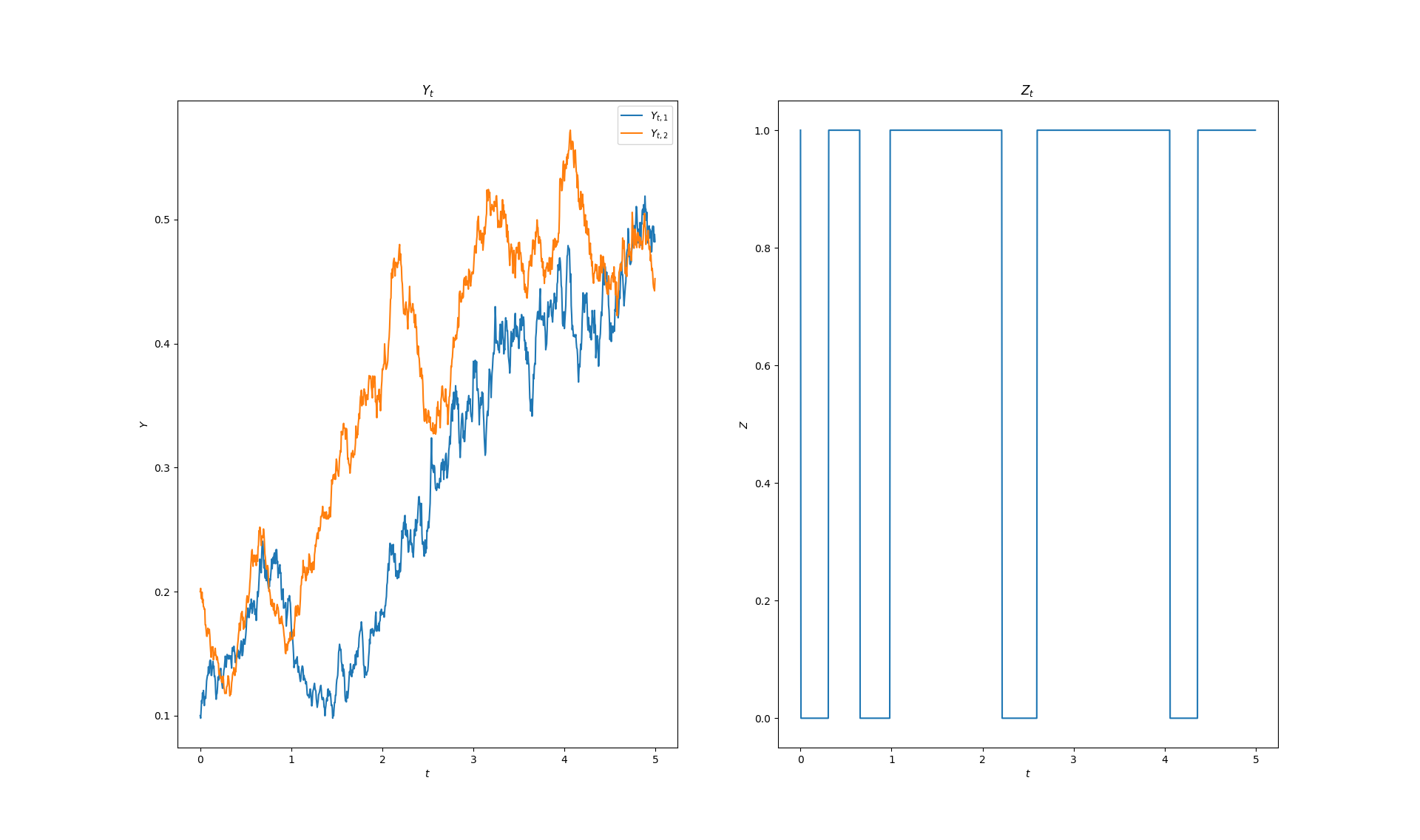}}
	\end{center}
	\caption{Realisations of the diffusion process $Y_t$ and the Markov chain $Z_t$.}
	\label{stoch_processes}
\end{figure}

We begin by constructing the curves for the energy futures model. We set the discount rate $r=0.1$ and choose for the starting values $u_0(e_1)=\begin{pmatrix} 0.9 \\ 0.6 \end{pmatrix}, u_0(e_2)=\begin{pmatrix} 0.3 \\ 0.2 \end{pmatrix}, c_0=\begin{pmatrix} 1 \\ 1.5 \end{pmatrix}$.
\begin{figure}[!h]
	\begin{center}
		\subfigure{\includegraphics[width=16cm, height=12cm]{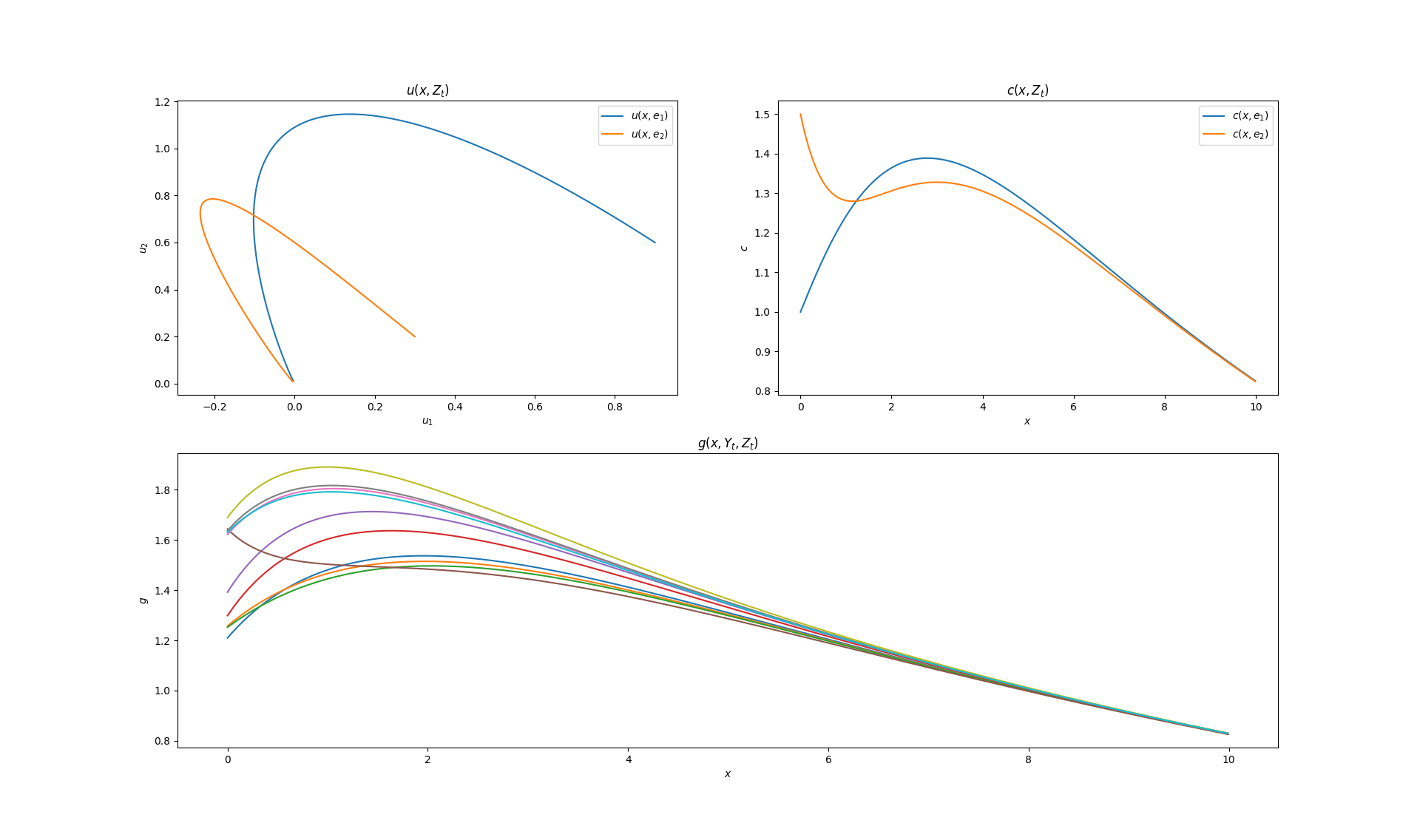}}
	\end{center}
	\caption{The components $u(x,Z_t)$ and $c(x,Z_t)$ of the term structure model of energy futures, along with the curve $g(x,Y_t,Z_t)$ plotted for times $t=\frac{j}{2}, j=1,...,10$.}
	\label{futures_2d}
\end{figure}

\begin{figure}[!h]
	\begin{center}
		\subfigure{\includegraphics[width=12cm, height =8.5cm]{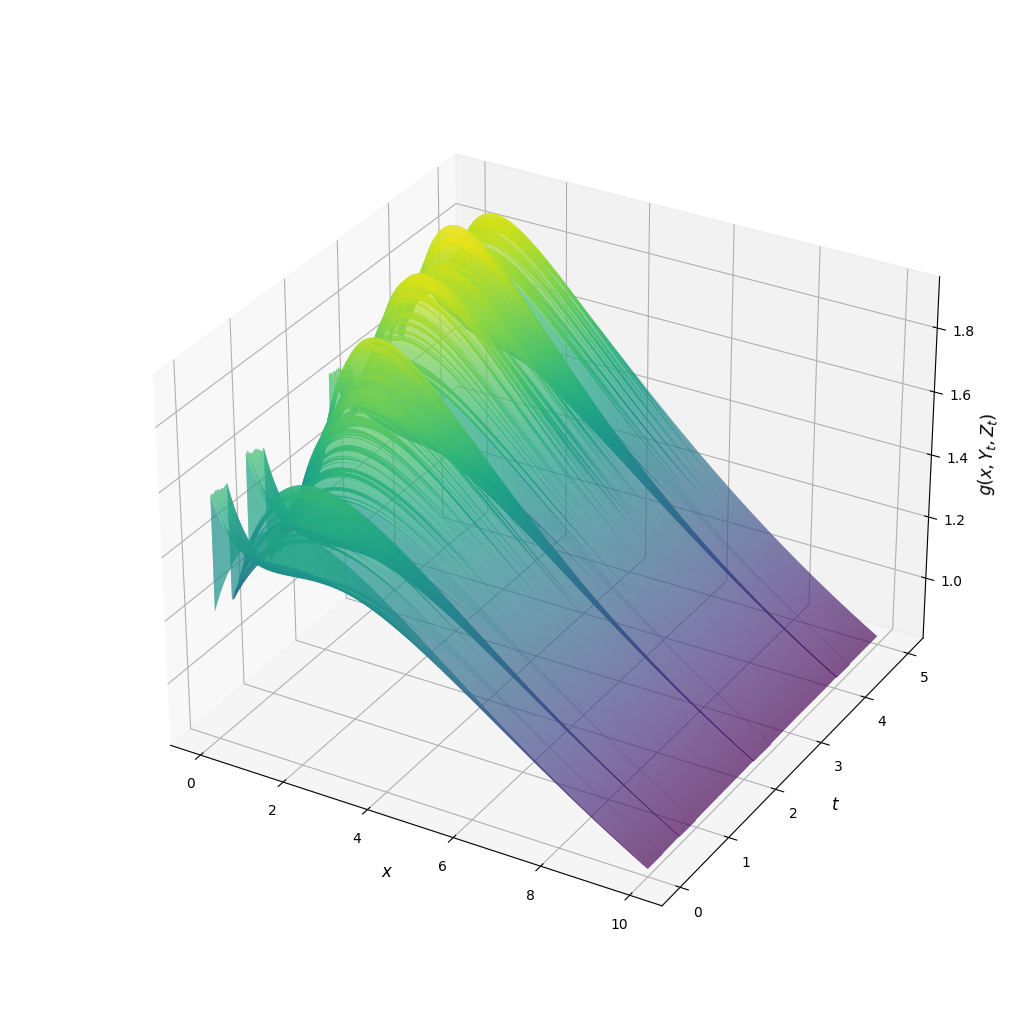}}
	\end{center}
	\caption{The term structure of energy futures and its evolution in time.}
	\label{futures_3d}
\end{figure}

For the interest rate model, we set the starting values $u_0=\begin{pmatrix} 0.9 \\ 0.6 \end{pmatrix}, c_0=\begin{pmatrix} 1 \\ 1.5 \end{pmatrix}$.
\begin{figure}[!h]
	\begin{center}
		\subfigure{\includegraphics[width=16cm, height=12cm]{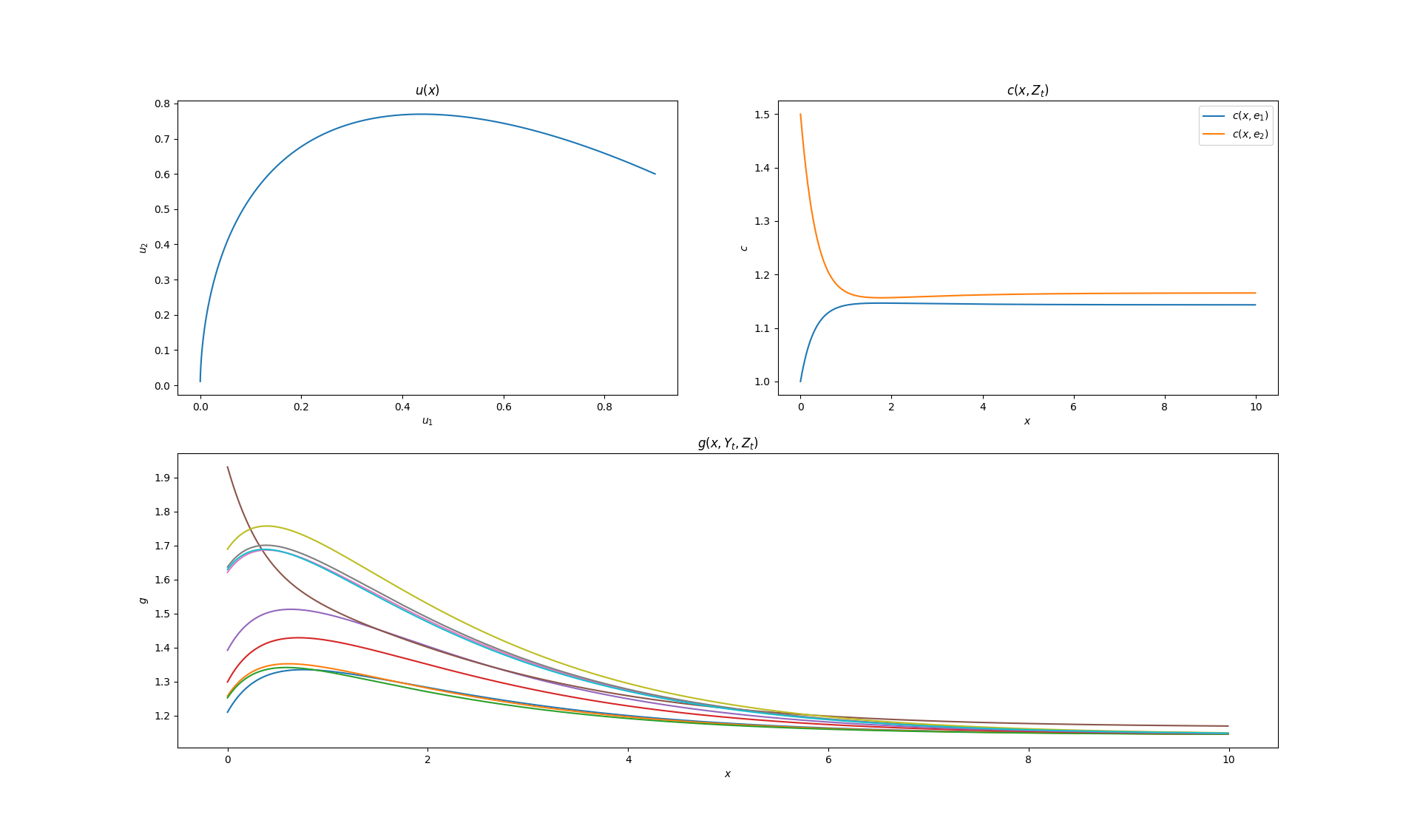}}
	\end{center}
	\caption{The components $u(x)$ and $c(x,Z_t)$ of the term structure model of interest rates, along with the curve $g(x,Y_t,Z_t)$ plotted for times $t=\frac{j}{2}, j=0,...,10$.}
	\label{bonds_2d}
\end{figure}
\begin{figure}[!h]
	\begin{center}
		\subfigure{\includegraphics[width =12cm, height =8.5cm]{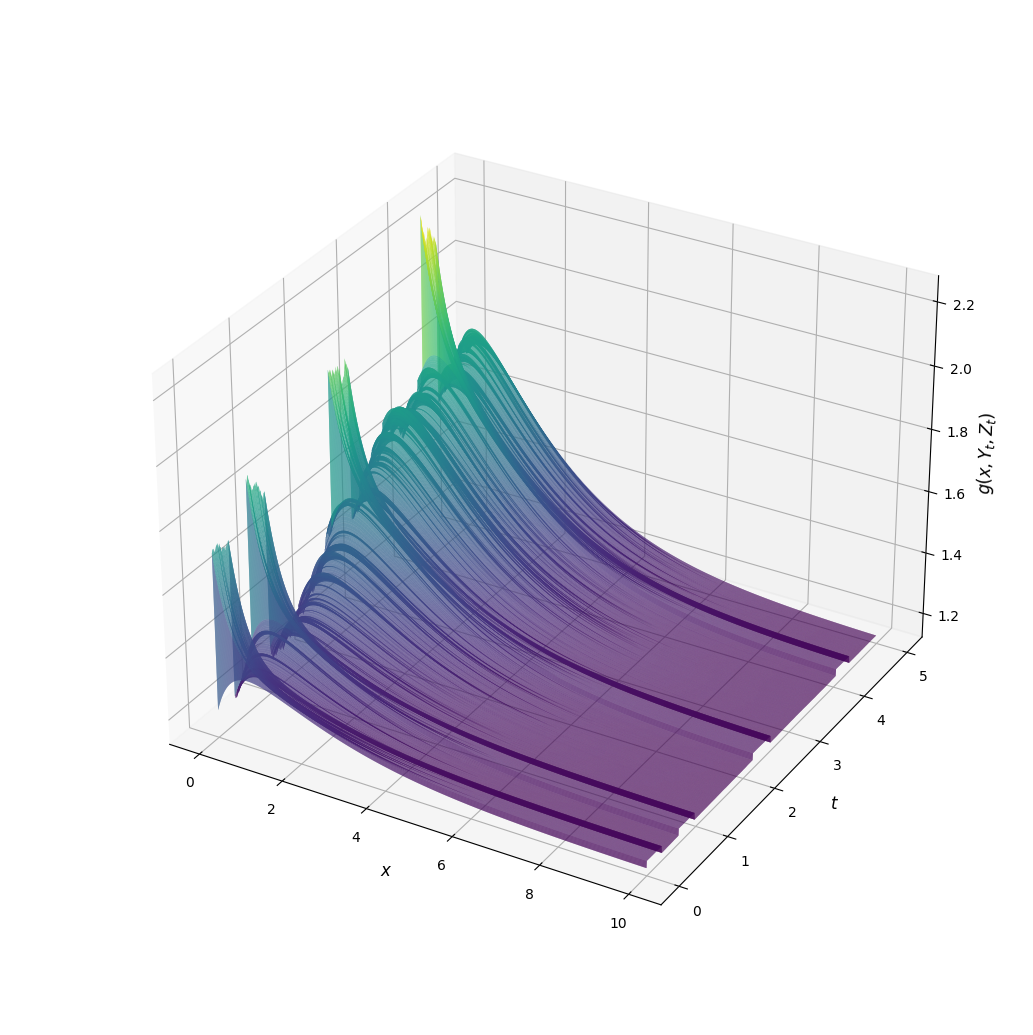}}
	\end{center}
	\caption{The term structure of interest rates and its evolution in time.}
	\label{bonds_3d}
\end{figure}

\clearpage
\appendix
\section{Technical tools}
In this appendix we collect some of the technical tools which are used throughout the paper. 
\begin{lemma}\label{a1}
	Let $Z$ be a finite-state Markov chain taking values in $E$ with generator $\mathcal{G}$, $Y_t=Y_0+\int _0^tb_sds+\int _0^t\sigma _sdW_s$ for some $\mathcal{F}_0$-measurable random variable $Y_0$ in $\mathbb{R}^d$ be an It\^o diffusion taking values in $\mathbb{R}^d$ and $g(\cdot ,e)\in C^{1,2}([0,T]\times\mathbb{R}^d,\mathbb{R})$ for each $e\in E$. Then the process
	\begin{equation*}
		M^g_t:=\sum _{s\leq t}\Delta g(x,Y_s,Z_s)-\int _0^t\mathcal{G}(g(x,Y_s,\cdot ))(Z_s)ds
	\end{equation*}
	is a local martingale.
\end{lemma}
\begin{proof}
        Let $\mu$ denote the jump measure of $Z$, $F$ be its jump law, and $\lambda$ denote its intensity. We can write the sum of the jumps using the jump measure, and the generator as a finite sum over the states in $E$. We therefore have
        \begin{align*}
                M^g_t&=\int _{[0,t]\times E}g(x,Y_s,z)-g(x,Y_s,Z_{s-})\mu (ds,dz)\\
                     &-\int _0^t\lambda (Z_s)\sum _{z\in E}g(x,Y_s,z)-g(x,Y_s,Z_{s-})F(Z_s,dz)ds\\
		     &=\int _{[0,t]}\sum _{z\in E}g(x,Y_s,z)-g(x,Y_s,Z_{s-})\mu (ds,dz)\\
                     &-\int _0^t\sum _{z\in E}g(x,Y_s,z)-g(x,Y_s,Z_{s-})\lambda (Z_s)F(Z_s,dz)ds.
        \end{align*}
        We observe here, that $\lambda (Z_t)F(Z_t,dz)dt=\nu(dt,dz)$, where $\nu$ is the compensator of the jump measure $\mu$. We therefore get
        \begin{align*}
		M^g_t&=\int _{[0,t]\times E}\left( g(x,Y_s,z)-g(x,Y_s,Z_{s-})\right) (\mu -\nu)(ds,dz).
        \end{align*}
        By \cite[Theorem II.1.8(ii)]{jacod} this is a local martingale.
\end{proof}
\begin{lemma}\label{a2}
	Let $(F_t(x))_{t\geq 0}$ be an $\mathbb{R}^d$-valued  semimartingale for any $x\in\mathbb{R}_+$ with dynamics of the form $F_t=F_0(x)+\int_0^t\beta _s(x)ds+\int _0^t\Sigma _s(x)dW_s+\sum _{s\leq t}\Delta F_{s-}(x)$ such that
        \begin{enumerate}
		\item $\Sigma (x)$ and $\beta (x)$ are progressively measurable for all $x\geq 0$.
		\item $(t, x, \omega)\mapsto \beta _t(x)(\omega)$ is $\mathcal{B}(\mathbb{R}_+)\otimes\mathcal{B}(\mathbb{R_+})\otimes\mathcal{F}$-measurable and $\int _0^T\int _0^T\vert\beta _s(x)\vert dxds <\infty$ for all $T>0$.
		\item $\sup _{x,t\leq T}\Vert\Sigma _t(x)\Vert <\infty$ for all $T>0$, where $\Vert\cdot\Vert$ denotes the matrix operator norm.
		\item $(t, x, \omega)\mapsto \Delta F _{t-}(x)(\omega)$ is $\mathcal{B}(\mathbb    {R}_+)\otimes\mathcal{B}(\mathbb{R_+})\otimes\mathcal{F}$-measurable and $\sum _{s\leq t}\vert\Delta F_{s-}(x)$ has only finitely many non-vanishing summands.
        \end{enumerate}
        Note that above all inequalities are considered $\omega$-wise. Define $X_t=\int _0^{T-t}F_t(x)dx$. Then $X$ fulfills
	\begin{equation}
		X_t=X_0+\int _0^t\int _0^{T-s}\beta _s(x)dxds+\int _0^t\int _0^{T-s}\Sigma _s(x)dxdW_s+\sum _{s\leq t}\int _0^{T-s}\Delta F_{s-}(x)dx-\int _0^tF_s(T-s)ds.
	\end{equation}
\end{lemma}
\begin{proof}
        We have, using the definition of $F_t(x)$ and the linearity of the integral,
        \begin{align*}
                X_t&=\int _0^{T-t}F_t(x)dx\\
		   &=\int _0^{T-t}\left( F_0(x)+\int _0^t\beta _s(x)ds+\int _0^t\Sigma _s(x)dW_s+\sum _{s\leq t}\Delta F_{s-}(x)\right) dx \\
                   &=\int _0^{T-t}F_0(x)dx+\int _0^{T-t}\int _0^t\beta _s(x)dsdx\\
		   &+\int _0^{T-t}\int_0^t\Sigma _s(x)dW_sdx+\int _0^{T-t}\sum _{s\leq t}\Delta F_{s-}(x)dx.
        \end{align*}
        In the next step, we make use of Fubini's theorem and split the integrals into two parts as follows
        \begin{align*}
                X_t&=\int _0^{T}F_0(x)dx+\int _0^t\int _0^{T-s}\beta _s(x)dxds\\
		   &+\int _0^t\int _0^{T-s}\Sigma _s(x)dxdW_s+\sum _{s\leq t}\int _0^{T-s}\Delta F_{s-}(x)dx\\
                   &-\int _{T-t}^TF_0(x)dx-\int _0^t\int _{T-t}^{T-s}\beta _s(x)gdxds\\
		   &+\int _0^t\int _{T-t}^{T-s}\Sigma _s(x)dxdW_s+\sum _{s\leq t}\int _{T-t}^{T-s}\Delta F_{s-}(x)dx.
        \end{align*}
        Using Fubini's theorem once more on the second part yields
        \begin{align*}
		X_t&=X_0+\int _0^t\int _0^{T-s}\beta _s(x)dxds+\int_0^t\int _0^{T-s}\Sigma _s(x)dxdW_s+\sum _{s\leq t}\int _0^{T-s}\Delta F_{s-}(x)dx\\
		   &-\int _{T-t}^T\left( F_0(x)+\int_0^{T-x}\beta _s(x)ds+\int _0^{T-x}\Sigma _s(x)dW_s+\sum _{0\leq T-x}\Delta F_{s-}(x)\right) dx.
        \end{align*}
	We have, by defintion of the process $F_t(x)$,
	\begin{equation*}
		\begin{aligned}
			\int _{T-t}^T\left( F_0(x)+\int_0^{T-x}\beta _s(x)ds+\int _0^{T-x}\Sigma _s(x)dW_s+\sum _{0\leq T-x}\Delta F_{s-}(x)\right) dx =\int _{T-t}^{T}F_{T-x}(x)dx.
		\end{aligned}
	\end{equation*}
        Using the variable substitution $s=T-x$, we obtain the desired result.

\end{proof}
\begin{lemma}\label{a3}
	The energy futures contract process $F(t,T_1,T_2)=\frac{1}{T_2-T_1}\int _{T_1-t}^{T_2-t}f_t(x)dx$ from Section \ref{sec:energy} fulfills 
        \begin{align*}
		F(t,T_1,T_2)&=F(0,T_1,T_2)+\frac{1}{T_2-T_1}\left(\int _0^t\int _{T_1-s}^{T_2-s}\left(\beta _s(x)-\partial _xg(x,Y_s,Z_s)\right) dxds\right.\\
			    &+\int _0^t\int _{T_1-s}^{T_2-s}\Sigma _s(x)dxdW_s+\sum _{s\leq t}\int _{T_1-s}^{T_2-s}\Delta g(x,Y_s,Z_s)dx\bigg).
        \end{align*}
\end{lemma}
\begin{proof}
	The claim follows by using Lemma \ref{a2} after writing $F(t,T_1,T_2)$ as the difference of two integrals starting at $0$ and the respective upper limits.
\end{proof}
\begin{lemma}\label{a4} 
	The zero-coupon bond price process $P(t,T)=\exp\left( -\int _0^{T-t}f_t(x)dx\right)$ from \ref{sec:bonds} satisfies
        \begin{align*}
		&P(t,T)-P(0,T)=\\
		&-\int _0^t P(s,T)\left( \int _0^{T-s}\beta _s(x)-\partial _xg(x,Y_s,Z_s)dx-\frac12\int _0^{T-s}\Sigma _s(x)dx\int _0^{T-s}\Sigma ^{\top}_s(x)dx+r_s\right) ds\\
		&-\int _0^tP(s,T)\int _0^{T-s}\Sigma _s(x)dxdW_s-\sum _{s\leq t}\Delta P(s,T).
        \end{align*}
\end{lemma}
\begin{proof}
	We set $X_t=\int _0^{T-t}f_t(x)dx$. The claim is obtained by using It\^o's formula on $P(t,T)=\exp (-X_t)$ and employing Lemma \ref{a2}.
\end{proof}

\begin{lemma}\label{a5}
	Let $V$ and $W$ be finite-dimensional vector spaces over $\mathbb{R}$ and let $f:V\rightarrow W$ be Lebesgue-measurable. Then $f$ is an affine function if and only if for all $x,v,w\in V$ and for all $\varepsilon >0$, we have for the second order difference quotient
	\begin{equation}\label{eq:A1}
		\begin{aligned}
			(D^2_{\varepsilon}f(x))(v,w):=\frac{f(x+\varepsilon (v+w))-f(x+\varepsilon w)-f(x+\varepsilon v)+f(x)}{\varepsilon ^2}=0.
		\end{aligned}
	\end{equation}
\end{lemma}
\begin{proof}
	The sufficiency is obvious. To show the necessity, we observe that a $f:V\rightarrow W$ is affine if and only if the function $T(x):=f(x)-f(0)$ is linear. Since $(D^2_{\varepsilon}f(x))(v,w)=0$ for all $x,v,w\in V$ and $\varepsilon >0$, we may choose $\varepsilon =1$ and $x=0$ and using the defintion \eqref{eq:A1} obtain
	\begin{equation*}
		0=f(v+w)-f(w)-f(v)+f(0)=T(v+w)-T(v)-T(w).
	\end{equation*}
	Therefore, $T$ is an additive function over $V$. Since $V$ is a finite-dimensional vector space of $\mathbb{R}$, and $f$ is measurable with respect to the Borel algebra over $V$, it follows from \cite[Theorem 5.5]{herrlich2006} that $T$ is linear, and therefore, $f$ is affine as asserted.
\end{proof}
\section{Proofs of preliminary results}
In this appendinx we collect the proofs of our preliminary results which are included for the reader's convenience.
\begin{proof}[Proof of Proposition \ref{prop:1}]
        An application of It\^o's product rule yields
        \begin{align*}
                d\tilde{F}(t,T_1,T_2)=e^{-rt}\left( dF(t,T_1,T_2)-rF(t,T_1,T_2)\right) .
        \end{align*}
        Since $Z_t$ is a a finite-state Markov chian, we have by using Lemma \ref{a1} and Fubini's theorem that the process 
        \begin{equation*}
                M^g_t:=\sum _{s\leq t}\int _{T_1-s}^{T_2-s}\Delta g(x,Y_t,Z_t)dx-\int _0^t\mathcal{G}\left(\int _{T_1-s}^{T_2-s}g(x,Y_s,\cdot )dx\right) (Z_s)ds
        \end{equation*}
        is a local martingale. By applying Lemma \ref{a3}, we may isolate the drift of the process $(\tilde{F}(t,T_1,T_2))_{t\in [0,T_1]}$ and observe that if Equation \eqref{eq:2} holds, then we have  
        \begin{equation}\label{eq:2a}
                \int _0^t\int _{T_1-s}^{T_2-s}\left(\beta _s(x)-\partial _xg(x,Y_s,Z_s)\right) dxds+\sum _{s\leq t}\int _{T_1-s}^{T_2-s}\Delta g(x,Y_s,Z_s)dx = M^g_t
        \end{equation}
        and therefore the process $(\tilde{F}(t,T_1,T_2))_{t\in [0,T_1]}$ is a local martingale for any $0\leq T_1<T_2<\infty$. On the other hand, if $(\tilde{F}(t,T_1,T_2))_{t\in [0,T_1]}$ is a local martingale for any $0\leq T_1<T_2<\infty$, we may again make use of Lemma \ref{a3} and note that Equation \eqref{eq:2a} has to hold. We may use the fact that the integrands have to agree almost everywhere and plug in the definition of the process $M_t^g$ to obtain the following drift condition
        \begin{align*}
                \int _{T_1-t}^{T_2-t}\bigg(\beta _t(x)-rg(x,Y_t,Z_t)-\partial _xg(x,Y_t,Z_t)\bigg) dx=-\mathcal{G}\left(\int _{T_1-t}^{T_2-t}g(x,Y_s,\cdot )dx\right) (Z_t).
        \end{align*}
        Because $g(\cdot ,\cdot, e_i)$ is of class $C^{1,2}$ for each $i=1,...,d$, all integrands are continuous functions in the variable $x$. We can therefore differentiate with respect to $T_2$. Setting $x=T_2-t$ and writing down the generator as a finite sum, we obtain
        \begin{align*}
                \beta _t(x)-rg(x,Y_t,Z_t)-\partial _xg(x,Y_t,Z_t)=-\lambda (Z_t)\sum _{j\in E}\left( g(x,Y_t,j)-g(x,Y_t,Z_t)\right) F(Z_t,j).
        \end{align*}
        Reordering terms, as well as using the fact that $\lambda (Z_t)F(Z_t,j)=q_{Z_t,j}$ and the definition of $\Delta f(j;t,x)$  yields $\eqref{eq:2}$.
\end{proof}

\begin{proof}[Proof of Proposition \ref{prop:2}] 
        Using the product rule, we get
        \begin{align*}
                d\tilde{P}(t,T)=e^{-\int _0^{t}r_sds}\left( dP(t,T)-r_tP(t,T)dt\right) .
        \end{align*} 
	Since $Z_t$ is a finite-state Markov chain, we obtain from Lemma \ref{a1} and from Fubini's theorem that the process
	\begin{equation*}
		M^g_t:=\sum _{s\leq t}\Delta e^{-\int _0^{T-s}g(x,Y_s,Z_s)dx}-\int _0^t\mathcal{G}\left(e^{-\int _0^{T-s}g(x,Y_s,\cdot )dx}\right) (Z_s) ds
	\end{equation*}
	is a local martingale. By applying Lemma \ref{a4}, we may isolate the drift of the process $(\tilde{P}(t,T))_{t\in [0,T]}$ and observe that if Equation \eqref{eq:4} holds, then we have
	\begin{equation}\label{eq:4a}
		\begin{aligned}
			&-\int _0^tP(s,T)\left(\int _0^{T-s}\beta _s(x)-\partial _xg(x,Y_s,Z_s)dx-\frac12\int _0^{T-s}\Sigma _s(x)dx\int _0^{T-s}\Sigma ^{\top}_s(x)dx\right) ds \\
			&+\sum _{s\leq t}\Delta P(s,T)=M^g_t,
		\end{aligned}
	\end{equation}
	and therefore the process $(\tilde{P}(t,T))_{t\in [0,T]}$ is a local martingale for any $0<T<\infty$. On the other hand, if $(\tilde{P}(t,T))_{t\in [0,T]}$ is a local martingale for any $0\leq T<\infty$, we may again make use of Lemma \ref{a4} and note that Equation \eqref{eq:4a} has to hold. We may use the fact that the integrands have to agree almost everywhere and plug in the definition of the process $M_t^g$ to obtain the following drift condition
        \begin{align*}
		&-P(t,T)\left( \int _0^{T-t} \beta _t(x)-\partial _xg(x,Y_t,Z_t)dx-\frac12\int _0^{T-t}\Sigma _t(x)dx \int _0^{T-t}\Sigma ^{\top}_t(x)dx\right)\\
                &=-\mathcal{G}\left(e^{-\int _0^{T-t}g(x,Y_t,\cdot )dx}\right) (Z_t).
        \end{align*}
        Writing out the generator in the case of a finite state Markov chain yields
        \begin{align*}
		&P(t,T)\left( \int _0^{T-t}\beta _t(x)-\partial _xg(x,Y_t,Z_t)dx-\frac12\int _0^{T-t}\Sigma _t(x)dx \int _0^{T-t}\Sigma ^{\top}_t(x)dx\right)\\
                &=\lambda (Z_t)\sum _{j\in E}\left( e^{-\int _0^{T-t}g(x,Y_t,j)dx}-e^{-\int _0^{T-t}g(x,Y_t,Z_t)dx}\right)F(Z_t,j).
        \end{align*}
        Dividing by $P(t,T)$ on both sides yields
        \begin{align*}
		&\int _0^{T-t}\beta _t(x)-\partial _xg(x,Y_t,Z_t)dx-\frac12\int _0^{T-t}\Sigma _t(x)dx\int _0^{T-t}\Sigma ^{\top}_t(x)dx\\
                &=\lambda (Z_t)\sum _{j\in E}\left( e^{\int _0^{T-t}g(x,Y_t,Z_t)-g(x,Y_t,j)dx}-1\right) F(Z_t,j).
        \end{align*}
	Now, since $g$ is of class $C^{1,2}$, the integrands are continous in the variable $x$ and we can differentiate both sides with respect to $T$. This yields
        \begin{align*}
		&\beta _t(T-t)-\partial _xg(T-t,Y_t,Z_t)-\Sigma _t(T-t)\int _0^{T-t}\Sigma ^{\top}_t(x)dx\\
                &=\lambda (Z_t)\sum _{j\in E}e^{\int _0^{T-t}g(x,Y_t,Z_t)-g(x,Y_t,j)dx}\left( g(T-t,Y_t,Z_t)-g(T-t,Y_t,j)\right) F(Z_t,j).
        \end{align*}
	After reordering terms, the claim follows by substituting $x=T-t$, renaming variables and using our definitions for the terms $\Delta f(j;t,x)$ and $\delta P(j;t,x)$, as well as the fact that $\lambda(Z_t)F(Z_t,j)=q_{Z_t,j}$.
\end{proof}
\newpage
\printbibliography

@article{karoui2000role,
  title={On the role of state variables in interest rates models},
  author={Karoui, Nicole El and Geman, Helyette and Lacoste, Vincent},
  journal={Applied Stochastic Models in Business and Industry},
  volume={16},
  number={3},
  pages={197--217},
  year={2000},
  publisher={Wiley Online Library}
}

@article{de1999dynamics,
  title={The dynamics of the forward interest rate curve: A formulation with state variables},
  author={De Jong, Frank and Santa-Clara, Pedro},
  journal={Journal of Financial and Quantitative Analysis},
  volume={34},
  number={1},
  pages={131--157},
  year={1999},
  publisher={Cambridge University Press}
}

@book{frachot1992factor,
  title={Factor analysis of the term structure: a probabilistic approach},
  author={Frachot, A and Janci, D and Lacoste, V},
  year={1992},
  publisher={Banque de France}
}

@article{huisman2003regime,
  title={Regime jumps in electricity prices},
  author={Huisman, Ronald and Mahieu, Ronald},
  journal={Energy economics},
  volume={25},
  number={5},
  pages={425--434},
  year={2003},
  publisher={Elsevier}
}

@article{ang2002regime,
  title={Regime switches in interest rates},
  author={Ang, Andrew and Bekaert, Geert},
  journal={Journal of Business \& Economic Statistics},
  volume={20},
  number={2},
  pages={163--182},
  year={2002},
  publisher={Taylor \& Francis}
}

@article{bansal2002term,
  title={Term structure of interest rates with regime shifts},
  author={Bansal, Ravi and Zhou, Hao},
  journal={The Journal of Finance},
  volume={57},
  number={5},
  pages={1997--2043},
  year={2002},
  publisher={Wiley Online Library}
}

@article{hamilton1989,
  title={A new approach to the economic analysis of nonstationary time series and the business cycle},
  author={Hamilton, James D},
  journal={Econometrica: Journal of the econometric society},
  pages={357--384},
  year={1989},
  publisher={JSTOR}
}

@article{hamilton1988,
  title={Rational-expectations econometric analysis of changes in regime: An investigation of the term structure of interest rates},
  author={Hamilton, James D},
  journal={Journal of Economic Dynamics and Control},
  volume={12},
  number={2-3},
  pages={385--423},
  year={1988},
  publisher={Elsevier}
}

@book{herrlich2006,
  title={Axiom of Choice},
  author={Herrlich, H.},
  isbn={9783540342687},
  lccn={2006921740},
  series={Lecture Notes in Mathematics},
  url={https://books.google.at/books?id=\_0cDCAAAQBAJ},
  year={2006},
  publisher={Springer Berlin Heidelberg}
}

@book{jacod,
	title={Limit Theorems for Stochastic Processes},
	author={Jacod, J. and Shiryaev, A.N.},
	isbn={9783540178828},
	lccn={lc87009865},
	series={Grundlehren der mathematischen Wissenschaften},
	url={https://books.google.at/books?id=sUgXKpUIdHwC},
	year={1987},
	publisher={Springer Berlin Heidelberg}
}

@Inbook{bjoerk,
author="Bj{\"o}rk, Tomas",
editor="Carmona, Ren{\'e} A.
and {\c{C}}inlar, Erhan
and Ekeland, Ivar
and Jouini, Elyes
and Scheinkman, Jos{\'e} A.
and Touzi, Nizar",
title="On the Geometry of Interest Rate Models",
bookTitle="Paris-Princeton Lectures on Mathematical Finance 2003",
year="2004",
publisher="Springer Berlin Heidelberg",
address="Berlin, Heidelberg",
pages="133--215",
isbn="978-3-540-44468-8",
doi="10.1007/978-3-540-44468-8_2",
url="https://doi.org/10.1007/978-3-540-44468-8_2"
}

@article{bjoerk1,
author = {Björk, Tomas and Christensen, Bent Jesper},
title = {Interest Rate Dynamics and Consistent Forward Rate Curves},
journal = {Mathematical Finance},
volume = {9},
number = {4},
pages = {323-348},
doi = {https://doi.org/10.1111/1467-9965.00072},
url = {https://onlinelibrary.wiley.com/doi/abs/10.1111/1467-9965.00072},
eprint = {https://onlinelibrary.wiley.com/doi/pdf/10.1111/1467-9965.00072},
year = {1999}
}

@TechReport{bjoerk2,
  author={Björk, Tomas and Landen, Camilla},
  title={{On the construction of finite dimensional realizations for nonlinear forward rate models}},
  year=2000,
  month=Dec,
  institution={Stockholm School of Economics},
  type={SSE/EFI Working Paper Series in Economics and Finance},
  url={https://ideas.repec.org/p/hhs/hastef/0420.html},
  number={420},
  doi={},
}

@TECHREPORT{bjoerk3,
title = {Finite dimensional Markovian realizations for stochastic volatility forward rate models},
author = {Björk, Tomas and Landén, Camilla and Svensson, Lars},
year = {2002},
institution = {Stockholm School of Economics},
type = {SSE/EFI Working Paper Series in Economics and Finance},
number = {498},
url = {https://EconPapers.repec.org/RePEc:hhs:hastef:0498}
}

@book{bjoerk4,
  title={Arbitrage Theory in Continuous Time},
  author={Bj{\"o}rk, T.},
  isbn={9780191533846},
  series={Oxford Finance Series},
  url={https://books.google.at/books?id=TJgjgJATeVIC},
  year={2004},
  publisher={Oxford University Press, Incorporated}
}

@article{bjoerk_svensson,
  title={{On the Existence of Finite-Dimensional Realizations for Nonlinear Forward Rate Models}},
  author={T. Bj{\"o}rk and Lars Svensson},
  journal={Mathematical Finance},
  year={2001},
  volume={11},
  pages={205-243}
}

@article{buehler,
author = {Buehler, Hans},
year = {2006},
month = {02},
pages = {178-203},
title = {Consistent Variance Curve Models},
volume = {10},
journal = {Finance and Stochastics},
doi = {10.1007/s00780-006-0008-2}
}

@article{cox,
 ISSN = {00129682, 14680262},
 URL = {http://www.jstor.org/stable/1911242},
 author = {John C. Cox and Jonathan E. Ingersoll and Stephen A. Ross},
 journal = {Econometrica},
 number = {2},
 pages = {385--407},
 publisher = {[Wiley, Econometric Society]},
 title = {A Theory of the Term Structure of Interest Rates},
 volume = {53},
 year = {1985}
}

@article{duffie,
author = {Duffie, Darrell and Kan, Rui},
title = {A Yield-Factor Model of Interest Rates},
journal = {Mathematical Finance},
volume = {6},
number = {4},
pages = {379-406},
doi = {https://doi.org/10.1111/j.1467-9965.1996.tb00123.x},
url = {https://onlinelibrary.wiley.com/doi/abs/10.1111/j.1467-9965.1996.tb00123.x},
eprint = {https://onlinelibrary.wiley.com/doi/pdf/10.1111/j.1467-9965.1996.tb00123.x},
year = {1996}
}

@Inbook{Elliott,
author="Elliott, Robert J.
and Wilson, Craig A.",
editor="Mamon, Rogemar S.
and Elliott, Robert J.",
title="The Term Structure of Interest Rates in a Hidden Markov Setting",
bookTitle="Hidden Markov Models in Finance",
year="2007",
publisher="Springer US",
address="Boston, MA",
pages="15--30",
isbn="978-0-387-71163-8",
doi="10.1007/0-387-71163-5_2",
url="https://doi.org/10.1007/0-387-71163-5_2"
}

@article{elliott_hjm,
author = {Robert J. Elliott and Tak Kuen Siu},
title = {Pricing regime-switching risk in an {HJM} interest rate environment},
journal = {Quantitative Finance},
volume = {16},
number = {12},
pages = {1791-1800},
year  = {2016},
publisher = {Routledge},
doi = {10.1080/14697688.2015.1136078},

URL = { 
        https://doi.org/10.1080/14697688.2015.1136078
    
},
eprint = { 
        https://doi.org/10.1080/14697688.2015.1136078
    
}
}

@article{elliott_hjm2,
author = {Elliott, Robert and Siu, Tak},
year = {2018},
month = {08},
pages = {45-59},
title = {The Heath-Jarrow-Morton Model with Regime Shifts and Jumps Priced: Fixed Income Modeling},
isbn = {978-3-319-95284-0},
journal = {Contributions to Management Science},
doi = {10.1007/978-3-319-95285-7_3}
}

@article{elhouar,
author = { Mikael   Elhouar },
title = {Finite‐dimensional {Realizations} of {Regime}‐switching {HJM} {Models}},
journal = {Applied Mathematical Finance},
volume = {15},
number = {4},
pages = {331-354},
year  = {2008},
publisher = {Routledge},
doi = {10.1080/13504860801987133},

URL = { 
        https://doi.org/10.1080/13504860801987133
    
},
eprint = { 
        https://doi.org/10.1080/13504860801987133
    
}
}

@PHDTHESIS{filipovic_phd,
	copyright = {In Copyright - Non-Commercial Use Permitted},
	year = {2000},
	author = {Filipović, Damir},
	size = {123 S.},
	language = {en},
	address = {Zürich},
	publisher = {ETH Zürich},
	DOI = {10.3929/ethz-a-003882242},
	title = {{Consistency problems for {HJM} interest rate models}},
	Note = {Diss. Mathematische Wissenschaften ETH Zürich, Nr. 13603, 2000.},
	school = {ETH Zurich}
}

@article{hull,
  title={{Pricing Interest-Rate-Derivative Securities}},
  author={J. Hull and A. White},
  journal={Review of Financial Studies},
  year={1990},
  volume={3},
  pages={573-592}
}

@book{hartshorne,
  title={Algebraic Geometry},
  author={Hartshorne, R.},
  isbn={9780387902449},
  lccn={lc77001177},
  series={Graduate Texts in Mathematics},
  url={https://books.google.at/books?id=3rtX9t-nnvwC},
  year={1977},
  publisher={Springer}
}

@article{jacod2,
author = {Çinlar, E. and Jacod, J. and Protter, Philip and Sharpe, M.},
year = {1980},
month = {01},
pages = {161-219},
title = {Semimartingales and Markov Processes},
volume = {54},
journal = {Probability Theory and Related Fields},
doi = {10.1007/BF00531446}
}

@book{kloeden,
  title={Numerical Solution of Stochastic Differential Equations},
  author={Kloeden, P.E. and Platen, E.},
  isbn={9783540540625},
  lccn={92015916},
  series={Stochastic Modelling and Applied Probability},
  url={https://books.google.at/books?id=BCvtssom1CMC},
  year={2011},
  publisher={Springer Berlin Heidelberg}
}

@phdthesis{Tappe,
author = {Tappe, Stefan},
title = {Finite dimensional realizations for term structure models driven by semimartingales},
school = {Humboldt-Universität zu Berlin, Mathematisch-Naturwissenschaftliche Fakultät II},
year = {2005},
doi = {http://dx.doi.org/10.18452/15369}
}

@article{tappe1,
author = {Stefan Tappe},
title = {{Flatness of invariant manifolds for stochastic partial differential equations driven by Lévy processes}},
volume = {20},
journal = {Electronic Communications in Probability},
number = {none},
publisher = {Institute of Mathematical Statistics and Bernoulli Society},
pages = {1 -- 11},
year = {2015},
doi = {10.1214/ECP.v20-3943},
URL = {https://doi.org/10.1214/ECP.v20-3943}
}

@article{tappe2,
author = {Tappe, Stefan },
title = {{Existence of affine realizations for stochastic partial differential equations driven by L\'evy processes}},
journal = {Proceedings of the Royal Society A: Mathematical, Physical and Engineering Sciences},
volume = {471},
number = {2178},
pages = {20150104},
year = {2015},
doi = {10.1098/rspa.2015.0104},

URL = {https://royalsocietypublishing.org/doi/abs/10.1098/rspa.2015.0104},
eprint = {https://royalsocietypublishing.org/doi/pdf/10.1098/rspa.2015.0104}
}

@article{tappe4,
 ISSN = {13645021},
 URL = {http://www.jstor.org/stable/20779291},
 author = {Stefan Tappe},
 journal = {Proceedings: Mathematical, Physical and Engineering Sciences},
 number = {2122},
 pages = {3033--3060},
 publisher = {The Royal Society},
 title = {An alternative approach on the existence of affine realizations for HJM term structure models},
 urldate = {2022-08-03},
 volume = {466},
 year = {2010}
}

@book{filipovic,
  title={{Term-Structure Models: A Graduate Course}},
  author={Filipovi\'c, D.},
  isbn={9783540680154},
  lccn={2009933038},
  series={Springer Finance},
  url={https://books.google.at/books?id=KqcSh6CavaAC},
  year={2009},
  publisher={Springer Berlin Heidelberg}
}

@article{filipovic2,
author = {Filipović, Damir},
year = {2001},
month = {07},
pages = {389-412},
title = {A General Characterization of One Factor Affine Term Structure Models},
volume = {5},
journal = {Finance and Stochastics},
doi = {10.1007/PL00013540}
}

@article{filipovic3,
title = {Stripping the Discount Curve - a Robust Machine Learning Approach},
author = {Filipović, Damir and Pelger, Markus and Ye, Ye},
year = {2022},
institution = {Swiss Finance Institute},
type = {Swiss Finance Institute Research Paper Series},
number = {22-24},
url = {https://EconPapers.repec.org/RePEc:chf:rpseri:rp2224}
}

@Article{hjm,
  author={Heath, David and Jarrow, Robert and Morton, Andrew},
  title={{Bond Pricing and the Term Structure of Interest Rates: A New Methodology for Contingent Claims Valuation}},
  journal={Econometrica},
  year=1992,
  volume={60},
  number={1},
  pages={77-105},
  month={January},
  doi={},
  url={https://ideas.repec.org/a/ecm/emetrp/v60y1992i1p77-105.html}
}

@article{musiela,
	author={Musiela, M.},
	title={{Stochastic PDEs and term structure models}},
	journal={Journees Internationales de Finance},
	publisher={IGR-AFFI},
	year={1993},
}

@article{filip,
	author = {Filipović, D. and Teichmann, J.},
	year = {2001},
	month = {07},
	journal = {arXiv: Probability},
	title = {{Finite dimensional Realizations of Stochastic Equations}}
}

@ARTICLE{teich,
    author = {Damir Filipović and Josef Teichmann},
    title = {On the geometry of the term structure of interest rates},
    journal = {Proceedings of The Royal Society of London. Series A. Mathematical, Physical and Engineering Sciences},
    year = {2004}
}

@article{teichmann1,
author = {Filipović, Damir and Teichmann, Josef},
year = {2003},
month = {02},
pages = {398-432},
title = {{Existence of Invariant Manifolds for Stochastic Equations in Infinite Dimension}},
volume = {197},
journal = {Journal of Functional Analysis},
doi = {10.1016/S0022-1236(03)00008-9}
}

@article{thorsten1,
author={Fontana, Claudio
and Grbac, Zorana
and G{\"u}mbel, Sandrine
and Schmidt, Thorsten},
title={Term structure modelling for multiple curves with stochastic discontinuities},
journal={Finance and Stochastics},
year={2020},
month={Apr},
day={01},
volume={24},
number={2},
pages={465-511},
issn={1432-1122},
doi={10.1007/s00780-020-00416-5},
url={https://doi.org/10.1007/s00780-020-00416-5}
}

@article{vasicek,
title = {An equilibrium characterization of the term structure},
journal = {Journal of Financial Economics},
volume = {5},
number = {2},
pages = {177-188},
year = {1977},
issn = {0304-405X},
doi = {https://doi.org/10.1016/0304-405X(77)90016-2},
url = {https://www.sciencedirect.com/science/article/pii/0304405X77900162},
author = {Oldrich Vasicek}
}

\end{document}